\documentclass[12pt]{article}

\usepackage{amsmath,amsthm,amssymb,amscd}

\usepackage{graphicx, float, epsfig}
\usepackage[hidelinks]{hyperref}
\usepackage{harvard}
\usepackage{rotating}
\usepackage{pww}

\usepackage{authblk}

\usepackage{enumitem} 
\usepackage{caption, subcaption} 
\usepackage{capt-of} 
\usepackage[flushleft]{threeparttable} 
\usepackage{multirow}
\usepackage{booktabs}

\usepackage{xcolor,soul} 

\newtheorem{theorem}{Theorem} 

\author[1]{Sheng Dai}
    \author[2]{Timo Kuosmanen}
    \author[1,3,\footnote{Corresponding author. \newline
    \hspace*{5mm} \textit{E-mail addresses:} \texttt{sheng.dai@aalto.fi (S. Dai)}, \texttt{timo.kuosmanen@utu.fi (T. Kuosmanen)},\\
    \hspace*{34mm} \texttt{xun.zhou@york.ac.uk (X. Zhou)}.}]{Xun Zhou}
    \affil[1~]{Aalto University School of Business}
    \affil[2~]{Turku School of Economics, University of Turku}
    \affil[3~]{Department of Environment and Geography, University of York}
\title{\bf Partial frontiers are not quantiles}

\date{May 2022}

\begin{document}

\citationmode{abbr}
\bibliographystyle{jbes}

\maketitle

\vfill
\vfill  

\begin{abstract}
\noindent
Quantile regression and partial frontier are two distinct approaches to nonparametric quantile frontier estimation. In this article, we demonstrate that partial frontiers are not quantiles. Both convex and nonconvex technologies are considered. To this end, we propose convexified order-$\alpha$ as an alternative to convex quantile regression (CQR) and convex expectile regression (CER), and two new nonconvex estimators: isotonic CQR and isotonic CER as alternatives to order-$\alpha$. A Monte Carlo study shows that the partial frontier estimators perform relatively poorly and even can violate the quantile property, particularly at low quantiles. In addition, the simulation evidence shows that the indirect expectile approach to estimating quantiles generally outperforms the direct quantile estimations. We further find that the convex estimators outperform their nonconvex counterparts owing to their global shape constraints. An illustration of those estimators is provided using a real-world dataset of U.S. electric power plants.  
\\[5mm]
\end{abstract}
\vfill

\thispagestyle{empty}

%

\newpage
\setcounter{page}{1}
\setcounter{footnote}{0}
\pagenumbering{arabic}
\baselineskip 20pt

%

\section{Introduction}\label{sec:intro}

The use of nonparametric quantile estimation as a tool to model and estimate production, cost, and distance functions has grown exponentially (see, inter alia, \citename{Aragon2005}, \citeyear*{Aragon2005}; \citename{Behr2010}, \citeyear*{Behr2010}; \citename{Wang2014c}, \citeyear*{Wang2014c}; \citename{Kuosmanen2015d}, \citeyear*{Kuosmanen2015d}; \citename{Jradi2019b}, \citeyear*{Jradi2019b}; \citename{Kuosmanen2020}, \citeyear*{Kuosmanen2020}; \citename{Dai2020}, \citeyear*{Dai2020}). The nonparametric quantile frontier estimation is popular among econometricians due to the fact that it provides an overall picture of the conditional distributions at any given quantiles and is more robust to the choice of directional vector, random noise, heteroscedasticity, and outliers in comparison with the conventional full frontier estimation. Furthermore, it allows more accurate estimation of shadow prices, which are essential for efficient environmental policy and management, by taking inefficiency explicitly into account (\citename{Kuosmanen2020b}, \citeyear*{Kuosmanen2020b}).

There exist two distinct approaches to nonparametric quantile frontier estimation: \textit{partial frontier} and \textit{quantile regression} (see Table \ref{tab:tab1}).\footnote{
In addition to the methods included in Table \ref{tab:tab1}, other nonparametric quantile estimators such as quantile smoothing splines (\citename{koenker1994}, \citeyear*{koenker1994}) could also be used to estimate the production function. For the sake of brevity, methods that are not commonly applied in the present context are excluded from Table \ref{tab:tab1}.  
}
Notable partial frontier estimators include stratified data envelopment analysis (stratified DEA) (\citename{Lovell1994}, \citeyear*{Lovell1994}), order-$m$ (\citename{Cazals2002a}, \citeyear*{Cazals2002a}), context-dependent DEA (\citename{Seiford2003}, \citeyear*{Seiford2003}), order-$\alpha$ (\citename{Aragon2005}, \citeyear*{Aragon2005}), smooth order-$\alpha$ (\citename{Martins-Filho2008}, \citeyear*{Martins-Filho2008}), and quantile-DEA (\citename{Atwood2020}, \citeyear*{Atwood2020}). In parallel, the quantile regression based frontier estimation literature dates back to stochastic DEA proposed by \citeasnoun{Banker1991},\footnote{
In a broader term, \citeasnoun{Aigner1976} construct expectiles to estimate production frontiers by replacing
$L_1$-norm distance by squared $L_2$-norm, which is the first work to apply expectile estimation. 
}
which is subsequently extended to concave nonparametric quantile regression by \citeasnoun{Wang2014c}. To ensure the uniqueness of quantile function estimates, \citeasnoun{Kuosmanen2015d} and \citeasnoun{Kuosmanen2020b} propose an indirect estimation through convex expectile regression (CER) by using an asymmetric least squares objective function. 

The distinctive features between the partial frontier and quantile regression estimators lie in the estimation strategy and in the interpretation of the error term $\varepsilon$. The partial frontier approach usually fits a frontier estimator to a subset of observations; in contrast, the quantile regression approach fits a quantile function using an asymmetric norm to the full sample of observations, similar to linear quantile regression (\citename{Koenker1978}, \citeyear*{Koenker1978}). Regarding the interpretation of the error term $\varepsilon$, the partial frontier studies usually assume away stochastic noise and attribute the deviation from the frontier to inefficiency alone (which implies that $\varepsilon_i \le 0$), whereas the quantile regression studies typically assume that the error term $\varepsilon_i$ is a composite error term that includes both inefficiency and noise. 
\begin{table}[H]
	\caption{Classification of nonparametric quantile-like estimators.}
	\vspace{-1em}
	\label{tab:tab1}
	\renewcommand\arraystretch{1.15}
	{\footnotesize
		\begin{center}
			\begin{tabular}{lll}
				\toprule
				&  \multicolumn{1}{l}{Nonconvex}    & \multicolumn{1}{l}{Convex} \\ 
				\midrule
    			\textit{Partial frontier}       &\textit{order-$\alpha$}    & \textit{Convexified order-$\alpha$}\\
				                                & \citeasnoun{Aragon2005}   & \citeasnoun{Ferreira2020}; \\
				                                & \textit{order-$m$}   & Section~\ref{sec:coa}\\ 
				                                & \citeasnoun{Cazals2002a}  &                      \\
				\textit{Quantile regression}    & \textit{Isotonic CQR}             &\textit{CQR}\\
				                                & Section~\ref{sec:icqr}    & \citeasnoun{Wang2014c} \\
				\textit{Expectile regression}   &\textit{Isotonic CER}              &\textit{CER}\\
				                                & Section~\ref{sec:icqr}    & \citeasnoun{Kuosmanen2015d};\\
				                                &                           & \citeasnoun{Kuosmanen2020b} \\
				\bottomrule
			\end{tabular}%
		\end{center}
	}
	\vspace{-1em}
\end{table}

While partial frontiers can be useful for more robust frontier estimation than the conventional full frontier approaches (e.g., DEA), they tend to perform poorly especially when estimating low quantiles, as will be demonstrated in this paper. Intuitively, partial frontiers rely on only a subset of observations and hence may suffer from small sample bias, especially for low quantiles. Moreover, order-$\alpha$, the most widely used partial frontier estimator, distinguishes itself from the quantile regression approach in that it cannot ensure the observed data is strictly split into proportions $\alpha$ below and $1-\alpha$ above for any $0< \alpha <1$. In fact, an $\alpha$-frontier guarantees a $100\alpha$\% chance of the observed data locating below the $\alpha$-frontier, but the realization of the $\alpha$-frontier becomes increasingly uncertain when $\alpha$ approaches zero.

As shown in Table~\ref{tab:tab1}, the nonparametric quantile-like estimators can also be classified into \textit{convex} and \textit{nonconvex} estimators depending on whether the assumption of convexity is imposed on the production possibility set. Under relaxed convexity assumptions, order-$\alpha$ is a representative nonconvex partial frontier estimator, which has been extensively applied in frontier estimation and performance evaluation (see, e.g., \citename{Wheelock2008}, \citeyear*{Wheelock2008}; \citename{Wheelock2009}, \citeyear*{Wheelock2009}; \citename{Matallin-Saez2019}, \citeyear*{Matallin-Saez2019}; \citename{Kounetas2021}, \citeyear*{Kounetas2021}). In sharp contrast, there are no nonconvex quantile estimators existing in the literature. To fill this gap, this paper proposes isotonic CQR and isotonic CER as alternatives to the nonconvex partial frontier estimators (to be discussed in Section \ref{sec:icqr}).

Thus far, the largest stream of partial frontier studies has developed separately from the more recent convex quantile/expectile regression studies. During their respective development, several interesting questions remain open. Considering the estimation strategy of the partial frontier approach (e.g., order-$\alpha$), it is likely to be robust at high quantiles (e.g., 90\% or 99\% quantiles; \citename{Wheelock2008}, \citeyear*{Wheelock2008}), which are of primary interest to production frontier estimation. However, what is the performance of the partial frontier approach at low quantiles (e.g., 5\% or 10\% quantiles)? How does the order-alpha estimator perform in comparison with the CQR and CER estimators and their nonconvex counterparts? Further, the distinction between the convex quantile and expectile regression has caused some confusion; thus, what are the differences and similarities between these two estimators and which one is more efficient in estimating quantile functions?

This paper contributes to the nonparametric quantile frontier literature in two ways. First, we demonstrate that partial frontiers are not quantiles. To this end, we extend the current estimation toolbox by proposing two new nonconvex estimators: isotonic CQR and isotonic CER. We then compare the finite sample performance of isotonic CQR/CER and order-$\alpha$ through Monte Carlo simulations and find that the order-$\alpha$ estimator performs relatively poorly and even can violate the quantile property, particularly at low quantiles. To ensure a fair and comprehensive comparison, we also develop convexified order-$\alpha$ as an alternative to CQR and CER. The Monte Carlo simulations show that the quantile regression approach outperforms the partial frontier approach as well in the convex case. 

Our second contribution shows that the indirect expectile approach to estimating quantiles generally outperforms the direct quantile estimations. Since there exists a one-to-one mapping between quantiles and expectiles and the estimated expectile functions are unique, the indirect estimation of quantiles using expectiles should be a more efficient approach. We then compare the performance of direct and indirect estimation of quantiles and find that the indirect expectile estimations work better than the direct quantile estimations in most scenarios.

In addition to the two main contributions above, we present a systematic classification of the quantile-like estimators and clarify their interpretations. The alternative estimators are illustrated in an empirical application, where the order-$\alpha$ estimator is found to violate the quantile and monotonicity properties. Furthermore, our simulation evidence confirms that the convex estimators (CQR/CER and convexified order-$\alpha$) outperform their nonconvex counterparts (isotonic CQR/CER and order-$\alpha$) owing to their global shape constraints.

The rest of this paper is organized as follows. Section~\ref{sec:meth} describes the convex quantile and expectile regression estimators and the partial frontier estimators. Section~\ref{sec:icqr} introduces the proposed isotonic convex quantile and expectile regression estimators. To illustrate and visualize the estimated quantile functions and partial frontiers, an empirical application to a dataset of U.S. electric power plants is presented in Section~\ref{sec:appli}. Section~\ref{sec:mc} performs a Monte Carlo study to compare the finite sample performance among nonparametric quantile frontier estimators. Section~\ref{sec:conc} concludes this paper with suggested avenues for future research. Formal proof and additional Monte Carlo simulation evidence are provided in Appendices~\ref{app:proof3} and \ref{app:experiments}.

%

\section{Models of quantile production function}\label{sec:meth}

\subsection{Quantile production function}
Suppose we observe input and output data $\{(\BX_i,Y_i)\}_{i=1}^n$, where $\BX \in \real^d$ is the $d$-dimensional input vector and $Y \in \real$ is the single output. Consider the following nonparametric regression model
\begin{equation}
\begin{aligned}
\label{eq:reg}
Y_i=f(\BX_i) + \varepsilon_i, \quad \mbox{ for } i = 1, \ldots, n,
\end{aligned}
\end{equation}	
where the regression function takes the form $f(x) = E(Y \,|\, \BX=x)$ and $\varepsilon_i$ is the error term satisfying $E(\varepsilon_i \,|\, \BX_i)=0$. The nonparametric model \eqref{eq:reg} does not assume any specific functional form for the regression function $f$, but rather assumes that $f$ satisfies certain axiomatic properties (e.g., monotonicity, concavity/convexity). As such, one can readily use this nonparametric model to characterize a production function by imposing shape constraints for all values of $x$ in the support of $\BX$ (see, e.g., \citename{Kuosmanen2008}, \citeyear*{Kuosmanen2008}; \citename{Kuosmanen2010a}, \citeyear*{Kuosmanen2010a}; \citename{Yagi2018}, \citeyear*{Yagi2018}).

In analogy to linear quantile regression, the nonparametric model \eqref{eq:reg} can be rephrased as a conditional nonparametric quantile function model. For any given quantile $\tau \in (0,1)$, the conditional nonparametric quantile function $Q_{Y_i}(\tau \,|\, \BX)$ is defined as
\begin{equation}
\begin{aligned}
\label{eq:qreg}
Q_{Y_i}(\tau \, | \, \BX_i)=f(\BX_i)+F_{\varepsilon_i}^{-1}(\tau),
\end{aligned}
\end{equation}	
where the quantile $\tau$ refers to that $Q_{Y_i}$ splits the observed data into proportions $\tau$ below and $1-\tau$ above; $F_{\varepsilon_i}$ is the cumulative distribution function of the error term $\varepsilon_i$. The estimation of $Q_{Y_i}$ is of central interest to nonparametric quantile regression, and there are a variety of estimators available in the literature as reviewed in Section~\ref{sec:intro}. In what follows we focus on the two distinct types of estimators (quantile regression and partial frontier) as well as their extensions.  

\subsection{Convex quantile regression}\label{sec:cqr}

By imposing monotonicity and global concavity on $f$, we can estimate the $\tau\textsuperscript{th}$ quantile production function \eqref{eq:qreg} by solving the following linear programming (LP) problem (\citename{Wang2014c}, \citeyear*{Wang2014c})\footnote{
In practice, problem \eqref{eq:cqr} can be solved by standard algorithms for LP such as CPLEX or MOSEK.}
\begin{alignat}{2}
 \underset{\alpha,\bbeta,\varepsilon^\text{+},\varepsilon^{-}}{\mathop{\min }}&\,\tau \sum\limits_{i=1}^{n}{\varepsilon _{i}^{+}}+(1-\tau )\sum\limits_{i=1}^{n}{\varepsilon _{i}^{-}}  &{}&  \label{eq:cqr}\\ 
\mbox{\textit{s.t.}}\quad
& y_i=\mathbf{\alpha}_i+ \bbeta_i^{'}\bx_i+\varepsilon_i^{+}-\varepsilon_i^{-} &\quad& \forall i \notag \\
& \alpha_i+\bbeta_i^{'}\bx_i \le \alpha_h+\bbeta_h^{'}\bx_i  &{}& \forall i,h \notag \\
& \bbeta_i\ge \bzero &{}& \forall i  \notag \\
& \varepsilon_i^{+}\ge 0,\ \varepsilon_i^{-} \ge 0 &{}& \forall i \notag
\end{alignat}
where the objective function is convex but not strictly convex on $\real^n$. Note that the error term $\varepsilon_i$ in \eqref{eq:reg} is now decomposed into two non-negative components $\varepsilon_i^+$ and $\varepsilon_i^-$ (i.e., $\varepsilon_i=\varepsilon_i^{+}-\varepsilon_i^{-}$). The first set of constraints can be interpreted as a multivariate regression equation. The second set of constraints, i.e., a system of Afriat inequalities, imposes concavity. The third set of constraints imposes monotonicity, and the last refers to the sign constraints on the decomposed error terms. 

Since it was proposed by \citeasnoun{Wang2014c}, convex quantile regression (CQR), as formulated in \eqref{eq:cqr}, has been applied to a number of studies because of its appealing features (e.g., \citename{Kuosmanen2015d}, \citeyear*{Kuosmanen2015d}; \citename{Jradi2019b}, \citeyear*{Jradi2019b}; \citename{Kuosmanen2020b}, \citeyear*{Kuosmanen2020b}). For example, the CQR estimator aims to estimate the conditional median or other quantiles of the response variable, and thus is more robust to random noise and heteroscedasticity than other central tendency estimators such as convex nonparametric least squares (\citename{Kuosmanen2008}, \citeyear*{Kuosmanen2008}) and penalized convex regression (\citename{Bertsimas2020}, \citeyear*{Bertsimas2020}). Furthermore, the CQR estimator is relatively computationally simple due to its LP formulation. Formally, the key properties of the CQR estimator are summarized in Theorem \ref{the:the1}.
\begin{theorem}
Let $\hat{\varepsilon}_i^{+}$ and $\hat{\varepsilon}_i^{-}$ be the optimal residuals estimated by CQR, $\hat{Q}_{y_i}(\tau \,|\, \bx)$ be the fitted values, and $n$ be the total number of observations.
\begin{enumerate}[label= \roman*), leftmargin=2\parindent]
    \setlength{\itemsep}{1pt}
    \setlength{\parskip}{0pt}
    \setlength{\parsep}{0pt}
    \item For any $\tau \in (0, 1)$, the number of strict positive residuals ($\hat{\varepsilon}_i^{+} > 0$) by $n_\tau^+$ and the number of strict negative residuals ($\hat{\varepsilon}_i^{-} > 0$) by $n_\tau^-$ always satisfy the inequalities: 
    \[
    \frac{n_\tau^+}{n} \le 1-\tau \quad \mbox{and} \quad \frac{n_\tau^-}{n} \le \tau.
    \]
    \item $\{\bx_i, y_i\}_{i=1}^\infty$are sequence of i.i.d. random variables generated by model \eqref{eq:reg}. With probability 1, asymptotically, $\frac{n_\tau^-}{n}$ equals $\tau$.
    \item In the optimal solution to problem \eqref{eq:cqr}, $\hat{Q}_{y_i}$ is not necessarily unique, even for the observed data points $\{\bx_i, y_i\}_{i=1}^n$.
\end{enumerate}
\label{the:the1}
\end{theorem}
\begin{proof}
See \citename{Wang2014c} (\citeyear*{Wang2014c}; Theorem 1) and \citename{Kuosmanen2020b} (\citeyear*{Kuosmanen2020b}; Propositions 1 and 2).
\end{proof}

One notable drawback of CQR is that the optimal solution to problem \eqref{eq:cqr} is not necessarily unique (see Theorem \ref{the:the1}), which also affects the estimated intercepts and slope coefficients (i.e., $\hat{\alpha}_i$ and $\hat{\beta}_{ij}$). This non-uniqueness problem of quantile regression could be safely assumed away if the regressors $\bx$ are randomly drawn from a continuous distribution. This, however, is often not the case in real applications, where two or more units may use exactly the same amount of inputs (i.e., $x_i = x_j$ for units $i$ and $j$). In production economics, the non-uniqueness of CQR emerges as a problem, particularly in samples where inputs $\bx$ are discrete variables (e.g., considering the number of employees in small firms).

\subsection{Convex expectile regression}\label{sec:cer}

To ensure unique estimates of the quantile functions, \citeasnoun{Kuosmanen2015d} propose an indirect estimation of quantiles through expectile regression (\citename{Newey1987}, \citeyear*{Newey1987}). Given an expectile $\tilde{\tau} \in (0, 1)$, we can estimate the quantile function \eqref{eq:qreg} indirectly by first solving the following quadratic programming problem
\begin{alignat}{2}
\underset{\alpha,\bbeta,\varepsilon^{+},\varepsilon^{-}}{\mathop{\min}}&\,\tilde{\tau} \sum\limits_{i=1}^n(\varepsilon _i^{+})^2+(1-\tilde{\tau} )\sum\limits_{i=1}^n(\varepsilon_i^{-})^2   &{}&  \label{eq:cer} \\ 
\mbox{\textit{s.t.}}\quad
& y_i=\mathbf{\alpha}_i+ \bbeta_i^{'}\bx_i+\varepsilon _i^{+}-\varepsilon _i^{-} &\quad& \forall i \notag \\
& \mathbf{\alpha}_i+\bbeta_{i}^{'}{{\bx}_{i}}\le \mathbf{\alpha}_h+\bbeta _h^{'}\bx_i  &{}& \forall i,h \notag \\
& \bbeta_i\ge \bzero &{}& \forall i \notag \\
& \varepsilon _i^{+}\ge 0,\ \varepsilon_i^{-} \ge 0 &{}& \forall i \notag
\end{alignat}
where the CER problem now minimizes the asymmetric squared deviations instead of the absolute deviations in \eqref{eq:cqr}. The quadratic objective function in \eqref{eq:cer} can guarantee the uniqueness of estimated quantile functions. The corresponding expectile property can be stated in Theorem \ref{the:the2}.
\begin{theorem}
Let $\hat{\varepsilon}_i^{+} $and $\hat{\varepsilon}_i^{-}$ be the optimal residuals estimated by CER. For any $\tilde{\tau} \in  (0, 1)$, we have
    \[
    \tilde{\tau}  = \frac{\sum\limits_{i=1}^n \hat{\varepsilon}_i^-}{\sum\limits_{i=1}^n \hat{\varepsilon}_i^+ + \sum\limits_{i=1}^n \hat{\varepsilon}_i^-}.
    \]
\label{the:the2}
\end{theorem}
\begin{proof}
See \citename{Kuosmanen2020b} (\citeyear*{Kuosmanen2020b}; Proposition 3).
\end{proof}

Compared with the quantile estimation, indirect expectile estimation can be more efficient due to the fact that asymmetric least squares makes use of the distance to observations instead of the discrete count of observations below or above the curve. In another context, the estimated expectile function has been suggested to be more sensitive to outliers than the estimated quantile function (\citename{Waltrup2015}, \citeyear*{Waltrup2015}; \citename{Daouia2020}, \citeyear*{Daouia2020}), which, however, is not supported by our Monte Carlo simulations (see Appendix~\ref{app:experiments}). 

Beyond those discrepancies between the direct and indirect estimation of quantiles, both approaches can be connected by a unique one-to-one mapping from quantile $\tau$ to expectile $\tilde{\tau}$. There exists a bijective function such that $m_{\tilde \tau} = q_{\tau}$, where expectile $\tilde{\tau}$ is defined as below (\citename{DeRossi2009}, \citeyear*{DeRossi2009})
\begin{equation*}
	\tilde{\tau} = \dfrac{\int_{-\infty}^{q_\tau}(y-q_\tau)dF(y)}{\int_{-\infty}^{q_\tau}(y-q_\tau)dF(y) - \int_{q_\tau}^{\infty}(y-q_\tau)dF(y)}, 
\end{equation*}
where $\int_{-\infty}^{q_\tau}(y-q_\tau)dF(y)$ and $\int_{q_\tau}^{\infty}(y-q_\tau)dF(y)$ are the lower and upper partial moments, respectively, and $F(y)$ is the cumulative distribution function of $y$. Therefore, we can always convert the expectile based quantile estimates $\hat{m}_{\tilde{\tau}}$ from the quantile estimates $\hat{q}_\tau$, and vice versa.

In practice, a simple procedure suggested by \citeasnoun{Efron1991} is first to estimate the expectile and then indirectly determine the corresponding quantile by counting the number of negative residuals $\varepsilon_i^-$ that take strictly positive values. More recently, \citeasnoun{Waltrup2015} propose a similar but more efficient approach by using the linear interpolation method. Note that all alternative transformation procedures rely on the quantile property stated in Theorem \ref{the:the1} (parts i and ii). 

However, the effectiveness of indirect estimation of quantiles through expectile regression has not been tested in the present context of CER. Moreover, as an alternative to the direct quantile regression, we really do not know about the finite sample performance of CER. In Section~\ref{sec:mc}, we will systematically compare the performance of these two approaches through Monte Carlo simulations.

Furthermore, while the estimated quantile function, $\hat{Q}_{y_i}$, is always unique in the CER estimation, the feasible set of problem \eqref{eq:cer} could be unbounded. That is, there may exist multiple combinations of shadow prices $\hat{\beta}_{ij}$) leading to the same optimal value of the objective function (\citename{Dai2021a}, \citeyear*{Dai2021a}). The non-unique estimates in both CQR and CER may further cause a longstanding problem of quantile crossing in quantile estimation (\citename{Dai2022}, \citeyear*{Dai2022}). Addressing the non-uniqueness estimation in both CQR and CER estimators is left as an interesting avenue for further research. 

\subsection{Order-\texorpdfstring{$\alpha$}{}}\label{sec:oa}

To ameliorate the sensitivity of DEA to outliers, \citeasnoun{Cazals2002a} propose the first robust nonparametric frontier estimator, order-$m$, where the estimated frontier is viewed as a ``trimmed” frontier. Subsequently, \citeasnoun{Aragon2005} develop a similar order-$\alpha$ estimator based on the conditional quantiles of an appropriate distribution. 

Consider the following standard production possibility set
\begin{equation*}
    \Ps = \{(x,y) \in \real_+^{d+1}\mid x\hbox{ can produce } y\}
\end{equation*}
where we assume that the set $\Ps$ is the support of the joint distribution of $(X, Y)$. For the partial frontier estimation, the focus is on the interior set 
\begin{equation*}
    \Ps^* = \{(x,y) \in \Ps \mid F_X(x) > 0\}
\end{equation*}
where $\Ps^* \subseteq \Ps$, $F$ is the joint distribution function of $(x, y)$; $F(x,y) = P(\{(X, Y): X \le x, Y \le y\})$, and $F_X(x)$ is the associated marginal distribution function of $x$. Given the level of inputs $x$, following \citeasnoun{Aragon2005} we can define the $\tau$ frontier as\footnote{
To be consistent with quantile estimators, the original notation of quantiles, $\alpha$, in the order-$\alpha$ estimator is replaced by the same-meaning notation $\tau$.
}
\begin{equation}
q_\tau(x):= F^{-1}(\tau\,|\,x)=\text{inf}\{y\ge 0 \mid F(y\,|\,x)\ge \tau\}
\label{eq:eq3}
\end{equation}
where $F(y\,|\,x) = F(x,y) \slash F_X(x)$ and it is the conditional distribution function of $Y$ given $X \le x$. Eq. \eqref{eq:eq3} indicates that the $\tau$ frontier falls below 100(1-$\tau$)\% of observations that use a smaller level of input than $x$. If the distribution function $F(y\,|\,x)$ is strictly increasing, then $q_\tau(x)= F^{-1}(\tau\,|\,x)$, where $F^{-1}(\tau\,|\,x)$ is the inverse of $F(y\,|\,x)$ similar to $F_\varepsilon^{-1}(\tau)$ in Eq.~\eqref{eq:qreg}. To estimate the partial frontier $q_\tau(x)$, one can obtain an empirical estimate by inverting the conditional empirical distribution function $\hat{F}(y\,|\,x)$
\begin{equation}
\hat{q}_{\tau, n}(x):= F^{-1}(\tau\,|\,x)=\text{inf}\{y \mid \hat{F}(y\,|\,x)\ge \tau\}
\label{eq:eq4}
\end{equation}

One interesting property of the order-$\alpha$ estimator is that as $\tau \rightarrow 1$, the function $q_\tau$ converges to the free disposal hull (FDH) full frontier, $q_1$, which is a monotone nondecreasing function. However, the function $q_\tau$ per se does not satisfy monotonicity unless we impose other assumptions (see Proposition 2.5 in \citename{Aragon2005}, \citeyear*{Aragon2005}). Further, FDH and order-$\alpha$ are step functions and hence cannot be used for shadow pricing.

In the subsequent literature, the order-$\alpha$ approach has been extended to incorporate the multivariate setting (\citename{Daouia2007}, \citeyear*{Daouia2007}), hyperbolic orientation (\citename{Wheelock2008}, \citeyear*{Wheelock2008}), and directional measures (\citename{Simar2012a}, \citeyear*{Simar2012a}). Meanwhile, order-$\alpha$ and its extensions have been widely applied in the context of productivity and efficiency analysis (see, e.g., \citename{kruger2012}, \citeyear*{kruger2012}; \citename{Wheelock2013}; \citeyear*{Wheelock2013}; \citename{carvalho2014}, \citeyear*{carvalho2014}). 

\subsection{Convexified order-\texorpdfstring{$\alpha$}{}}\label{sec:coa}

The order-$\alpha$ approach in \citeasnoun{Aragon2005} and \citeasnoun{Daouia2007} only assumes that the production function $f$ is monotone increasing in $x$. Since the additional concavity assumption of $f$ is commonly imposed in the production economic literature, we here impose convexity on order-$\alpha$ in line with the convexified order-$m$ approach (\citename{Daraio2007}, \citeyear*{Daraio2007}).\footnote{
\citeasnoun{Ferreira2020} propose another convexified version of order-$\alpha$, which also assumes virtual weight restrictions and non-variable returns to scale.
}

Our proposed convexified order-$\alpha$ estimator consists of a two-step estimation procedure: 1) we utilize order-$\alpha$ to estimate the order-$\alpha$ production frontier; 2) we apply the standard DEA-VRS (variable returns to scale) estimator to the estimated output on the order-$\alpha$ production frontier ($\hat{y}_{i}^{\tau, \text{order-}\alpha}$) and obtain the convexified order-$\alpha$ production frontier. Formally, the second step can be stated as 
\begin{alignat}{2}
\hat{\theta}_i^{\tau} = \underset{\theta, \lambda}{\text{max}}\bigg\{\theta \biggm| \theta \hat{y}_{i}^{\tau, \text{order-}\alpha} \le \sum_{j=1}^{n}\lambda_j \hat{y}_{j}^{\tau, \text{order-}\alpha};\, x_i \ge \sum_{j=1}^{n}\lambda_j x_j; &{}& \label{eq:eq5} \\ 
\sum_{j=1}^n \lambda_j =1;\, \lambda_j \ge 0 \,\, \forall j=1,\ldots,n\bigg\}  \notag
\end{alignat}
where $\hat{\theta}_i^{\tau}$ is the efficiency score and $\lambda_j$ is the intensity variable. Multiplying $\hat{y}_{i}^{\tau, \text{order-}\alpha}$ by $\hat{\theta}_i^{\tau}$ yields the estimated output $\hat{y}_{i}^{\tau}$ on the convexified order-$\alpha$ frontier. Note that in contrast to the original order-$\alpha$ estimator, the convexified order-$\alpha$ estimator satisfies the monotonicity and concavity of the production function by construction, which can be useful for shadow pricing non-market inputs and/or outputs. 

%

\section{Isotonic convex quantile regression}\label{sec:icqr}

Since the current methodological toolbox does not include a nonconvex quantile regression method in parallel with the order-$\alpha$ estimator, to enable performance comparison we need to extend the quantile function approach by relaxing the convexity assumption and relying on the monotonicity assumption only. 

Isotonic regression has a long history in statistics (see, e.g., \citename{Brunk1955}, \citeyear*{Brunk1955}; \citename{ayer1955}, \citeyear*{ayer1955}). It can be easily extended to a setup where the predictors can take values in any space with a partial order. Isotonic quantile regression (\citename{casady1976}, \citeyear*{casady1976}) is also well established to estimate the monotonic quantile function. In this section, we develop an alternative formulation of isotonic quantile regression, which is computationally convenient and more closely related to the convex CQR. We further utilize the new formulation to relax the convexity assumption of CER. 

Let $\chi:=\{\bx_i \in \real_+^{d}\}$ be a non-empty set with $d$ distinct elements in a metric space with a partial order, which is reflexive ($\bx_i \preccurlyeq \bx_i$, $\forall \bx_i \in \chi$), transitive (for $\bx_i, \bx_j, \bx_k \in \chi, \bx_i \preccurlyeq \bx_j$ and $\bx_j \preccurlyeq \bx_k$ imply $\bx_i \preccurlyeq \bx_k$), and antisymmetric (for $\bx_i, \bx_j \in \chi, \bx_i \preccurlyeq \bx_j$ and $\bx_j \preccurlyeq \bx_i$ imply $\bx_i = \bx_j$). Consider the production function $f$ is isotonic with respect to a partial ordering on $\chi$: if for any pair $\bx_i$, $\bx_h$ $\in \chi$, $\bx_i \preccurlyeq \bx_h$, then the fitted production function $f^*(\bx_i) \in \Mc$, where
\begin{equation*}
    \Mc:= \{ f \in \real^d: f(\bx_i) \le f(\bx_h)\}.
\end{equation*}

If the partial ordering is defined as the dominance relation (i.e., $\bx_i \preccurlyeq \bx_j$ if $\bx_i \le \bx_j$), then the non-decreasing production function satisfies monotonicity (i.e., free disposability of inputs); that is, isotonicity is equivalent to monotonicity. However, the partial ordering could also be defined by other criteria (e.g., revealed preference information), where isotonicity is not exactly the same as monotonicity. In this paper, we follow the general isotonic notation given above but note that monotonicity is an important special case of isotonicity.

For a given set of data $\{(\bx_i, y_i)\}_{i=1}^n$ and quantile $\tau$, convex quantile regression over the class $\Mc$ is
\begin{equation}
\hat{Q}(\tau \, | \, \bx_i) \in \operatorname*{arg\,min}_{f_\tau \in \Mc}\sum^{n}_{i=1}(y_i - f_\tau(\bx_i))(\tau - \textbf{1}\{y_i \le f_\tau(\bx_i)\})
\label{eq:iqr}
\end{equation}
where the isotonic CQR problem \eqref{eq:iqr} selects the best-fit isotonic quantile function from the class $\Mc$. In practice, however, it is impossible to directly search for the optimal solution from this infinite problem. Following \citeasnoun{Barlow1972}, we can harmlessly replace the class of isotonic quantile functions $\Ms$ by the following step functions $\Gc$
\begin{equation*}
	\Gc = \big \{Q: \real_+^d \rightarrow \real_+ \, | \, Q(\tau \, | \, \bx) = \sum\limits_{i=1}^n\delta_i \Zc (\tau \, | \, \bx_i) \big\}
\end{equation*}
where $\Zc(\tau \, | \, \bx_i)$ is an indicator function at a given quantile $\tau$ and is formulated as
\[
 \Zc(\tau \, | \, \bx_i) = 
  \begin{cases} 
   1       & \text{if } \bx_i \preccurlyeq \bx, \\
   0       & \text{otherwise}.
  \end{cases}
\]
and $\delta_i>0$ is the parameter to characterize the step height. Note that the step functions $\Gc$ are a subset of the isotonic functions $\Mc$ (i.e., $\Gc \subset \Ms$), which helps to transform the infinite problem \eqref{eq:iqr} to a finite problem (see, e.g., \citename{Barlow1972}, \citeyear*{Barlow1972}; \citename{Keshvari2013}, \citeyear*{Keshvari2013}).

Hence, problem \eqref{eq:iqr} can be solved through the following finite dimensional isotonic CQR problem\footnote{
Note that as an extension of CQR, isotonic CQR remains in the class of convex regression methods, even though the resulting step function is typically neither convex nor concave. Note also that the estimated step function envelops a union of $n$ convex sets.
}
\begin{alignat}{2}
 \underset{\mathbf{\alpha},\mathbf{\beta },{\mathbf{\varepsilon }}^{\text{+}},{\mathbf{\varepsilon }}^{-}}{\mathop{\min}}&\,\tau \sum\limits_{i=1}^{n}{\varepsilon _{i}^{+}}+(1-\tau )\sum\limits_{i=1}^{n}{\varepsilon _{i}^{-}}  &{}& \label{eq:icqr} \\ 
\mbox{\textit{s.t.}}\quad
& y_i=\mathbf{\alpha}_i+\bbeta _i^{'}\bx_i+\varepsilon_i^{+}-\varepsilon_i^{-} &\quad& \forall i \notag \\
& p_{ih}\Big(\mathbf{\alpha}_i+\bbeta _{i}^{'}{{\bx}_{i}}\Big)\le p_{ih}\Big(\mathbf{\alpha}_h+\bbeta _h^{'}\bx_i\Big)  &{}& \forall i,h \notag \\
& \bbeta_i\ge \bzero &{}& \forall i \notag \\
& \varepsilon _i^{+}\ge 0,\ \varepsilon_i^{-} \ge 0 &{}& \forall i \notag
\end{alignat}
where isotonic CQR requires an additional preprocessing step to determine the value of $p_{ih}$ that represents the partial order between observation $i$ and $h$. If $p_{ih}=0$, the concavity constraint on the production function $f$ is relaxed in isotonic CQR; otherwise, the isotonic CQR estimator is reduced to the original CQR estimator \eqref{eq:cqr}. Therefore, the isotonic CQR estimator provides an alternative way to model the class of nonparametric isotonic quantile regressions, which is computationally convenient and provides a clear link to CQR. 

To determine the value of $p_{ih}$ in \eqref{eq:icqr}, we need to define a binary matrix $\BP=\big[p_{ih} \big]_{n \times n}$
\[
 p_{ih} = 
  \begin{cases} 
   1       & \text{if } \bx_i \preccurlyeq \bx_h, \\
   0       & \text{otherwise}.
  \end{cases}
\]
The matrix $\BP$ converts the partial order relations between two observations into binary values and the value of $p_{ih}$ is determined by the standard dominance relations, which can be simply detected by an enumeration procedure suggested by \citeasnoun{Keshvari2013}. Further, the matrix $\BP$ can be interpreted as a preference matrix if the partial ordering denotes the preference of a decision maker. 

Similarly, we can replace the objective function of \eqref{eq:icqr} with the following quadratic objective function to guarantee the unique expectile estimation and derive the isotonic CER estimator
\vspace{-1em}
\begin{alignat}{2}
 \underset{\alpha, \bbeta, \varepsilon^+, \varepsilon^-}{\mathop{\min}}&\,\tilde{\tau} \sum\limits_{i=1}^n(\varepsilon _i^{+})^2+(1-\tilde{\tau} )\sum\limits_{i=1}^n(\varepsilon_i^{-})^2  &{}& \label{eq:icer} \\ 
\mbox{\textit{s.t.}}\quad
& y_i=\mathbf{\alpha}_i+\bbeta _i^{'}\bx_i+\varepsilon_i^{+}-\varepsilon_i^{-} &\quad& \forall i \notag \\
& p_{ih}\Big(\mathbf{\alpha}_i+\bbeta _{i}^{'}{{\bx}_{i}}\Big)\le p_{ih}\Big(\mathbf{\alpha}_h+\bbeta _h^{'}\bx_i\Big)  &{}& \forall i,h \notag \\
& \bbeta_i\ge \bzero &{}& \forall i \notag \\
& \varepsilon _i^{+}\ge 0,\ \varepsilon_i^{-} \ge 0 &{}& \forall i \notag
\end{alignat}

For both CQR and CER, the proposed isotonic CQR and isotonic CER can be interpreted as their nonconvex counterparts. Notably, both isotonic CQR and isotonic CER inherit the quantile and expectile properties from their parents, as Theorem \ref{the:the3} demonstrates.
\begin{theorem}
The quantile and expectile properties, i.e., Theorem \ref{the:the1} (part i) and Theorem \ref{the:the2}, are retained in isotonic CQR and isotonic CER, respectively.  
\label{the:the3}
\end{theorem}
\begin{proof}
See Appendix \ref{app:proof3}.
\end{proof}

The nonconvex nonparametric quantile regression estimators developed in this section enable us to compare the finite sample performance of the quantile and expectile regression approaches with the partial frontier approach. In the convex case, the nonparametric quantile regression estimator can be compared with the developed convexified partial frontier estimator.

%

\section{Empirical illustration of quantile functions}\label{sec:appli}

To gain an intuition of what the alternative quantile functions and partial frontiers look like, we first illustrate those estimators with a real cross-sectional dataset used in \citeasnoun{Kuosmanen2020b} and \citeasnoun{Dai2022}. It covers plant-level data on 130 U.S. electric power plants in 2014. A very similar dataset has been repeatedly used in the empirical demonstration of newly developed frontier estimators (see, e.g., \citename{Greene1990}, \citeyear*{Greene1990}; \citename{Gijbels1999}, \citeyear*{Gijbels1999}; \citename{Martins-Filho2008}, \citeyear*{Martins-Filho2008}).

Following \citeasnoun{Gijbels1999} and \citeasnoun{Martins-Filho2008}, we consider a univariate case where the output $y = \ln(Q)$ with $Q$ being the net generation of each power plant and the input $x = \ln(C)$ with $C$ being the sum of fixed cost and variable cost of electricity production. See \citeasnoun{Kuosmanen2020b} for a detailed discussion of the data sources and descriptive statistics.

Since there exists a one-to-one mapping between quantiles and expectiles, we estimate a number of expectiles (i.e., $\tilde{\tau}$ = 0.001, 0.002, $\ldots$, 0.999) and then determine the corresponding quantile $\tau$ by counting the number of negative residuals $\varepsilon_i$ that take strictly positive values (\citename{Efron1991}, \citeyear*{Efron1991}). Figure~\ref{fig:fig1} depicts the estimated monotonic quantile and expectile functions by order-$\alpha$, isotonic CQR, and isotonic CER at $\tau = 0.9$, $0.7$, $0.5$, and $0.3$, respectively.
\begin{figure}[H]
	\centering
	\begin{subfigure}[b]{0.495\textwidth}
		\centering
		\includegraphics[width=1\textwidth]{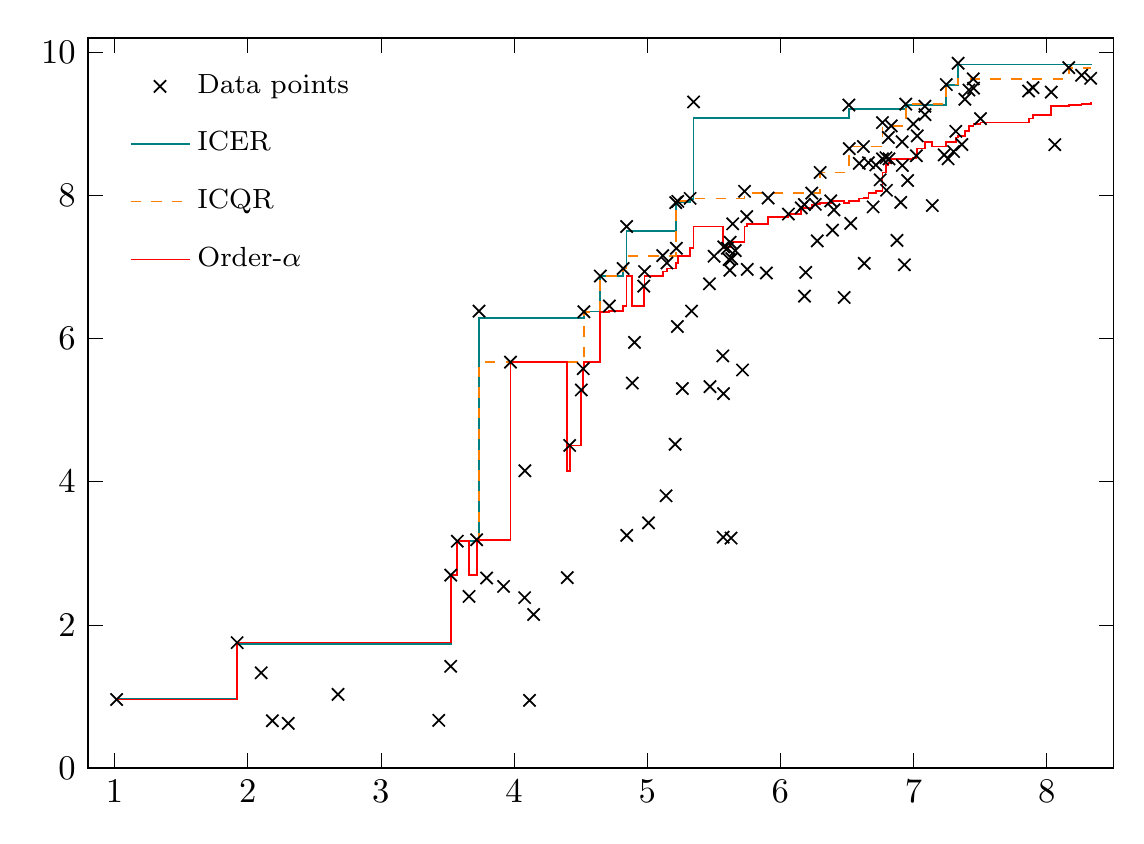} 
		\caption[]%
		{{\small $\tau=0.9$}}    
		\label{fig1:a}
	\end{subfigure}
	\begin{subfigure}[b]{0.495\textwidth}  
		\centering 
		\includegraphics[width=1\textwidth]{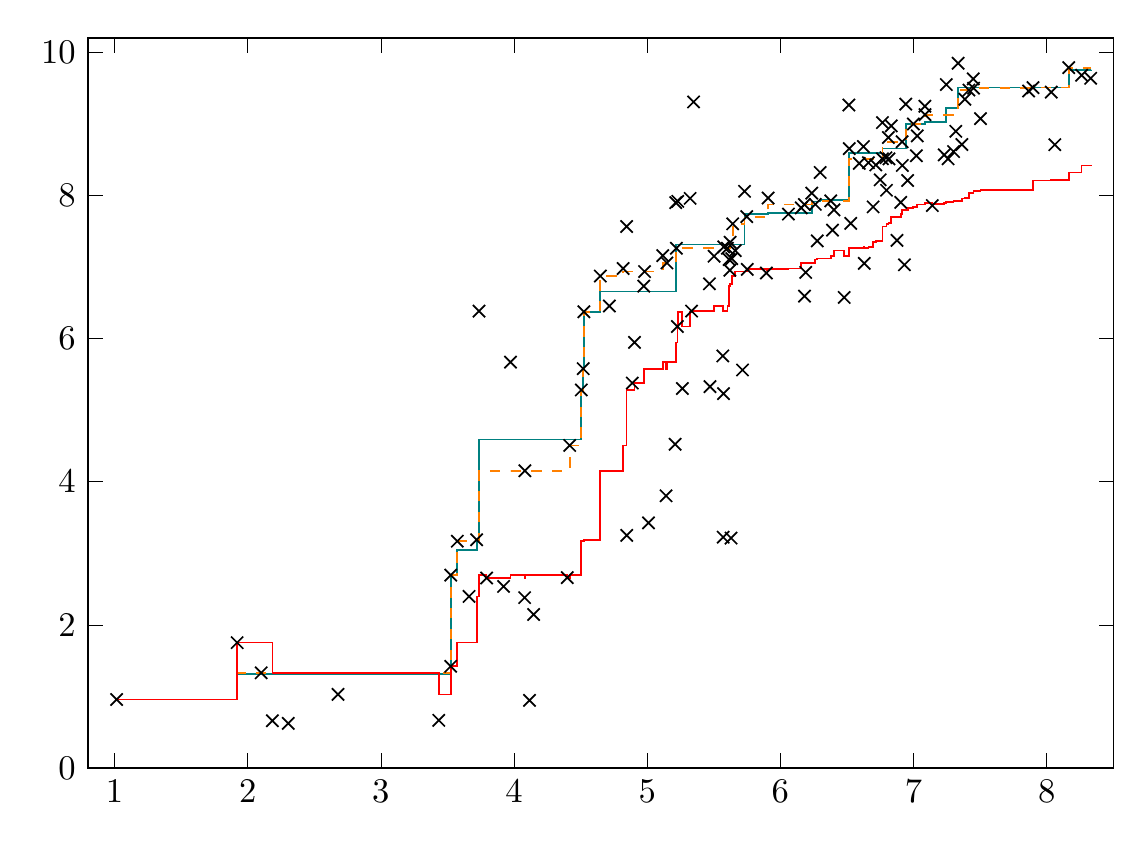}
		\caption[]%
		{{\small $\tau=0.7$}}    
		\label{fig1:b}
	\end{subfigure}
	\begin{subfigure}[b]{0.495\textwidth}   
		\centering 
		\includegraphics[width=1\textwidth]{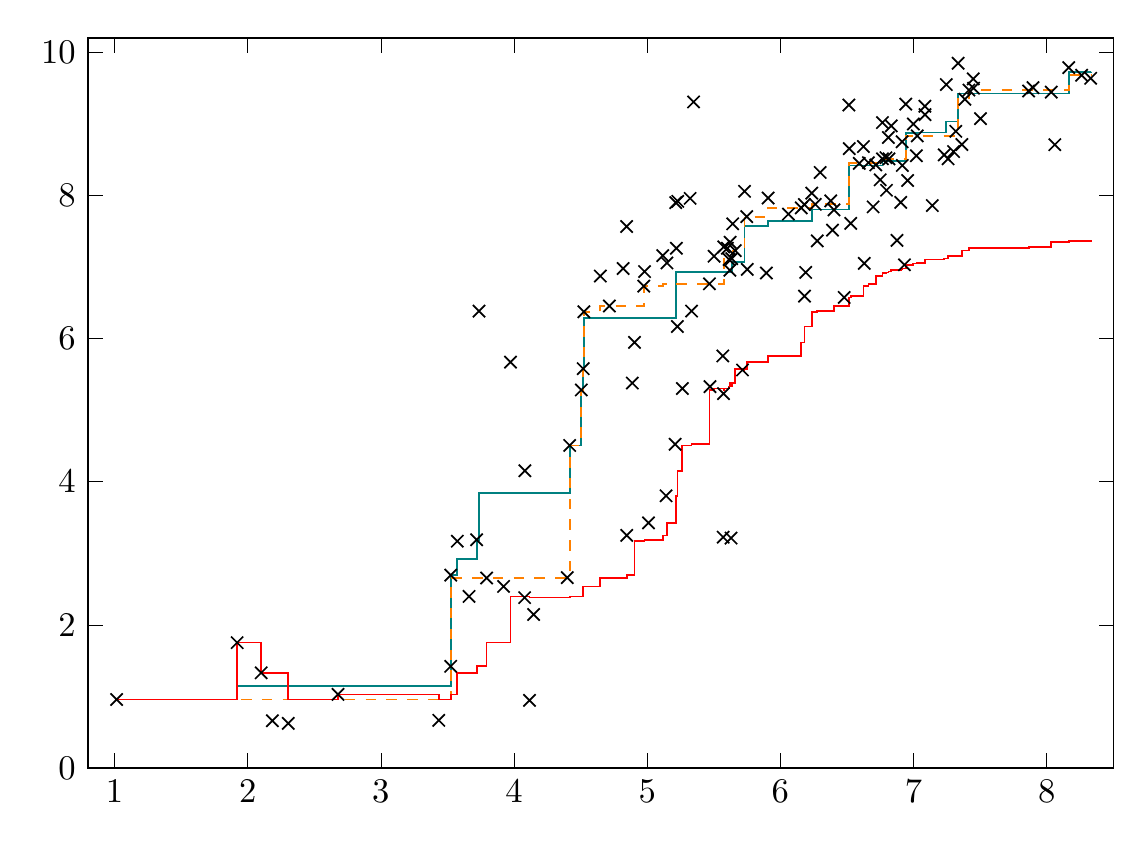} 
		\caption[]%
		{{\small $\tau=0.5$}}    
		\label{fig1:c}
	\end{subfigure}
	\begin{subfigure}[b]{0.495\textwidth}   
		\centering 
		\includegraphics[width=1\textwidth]{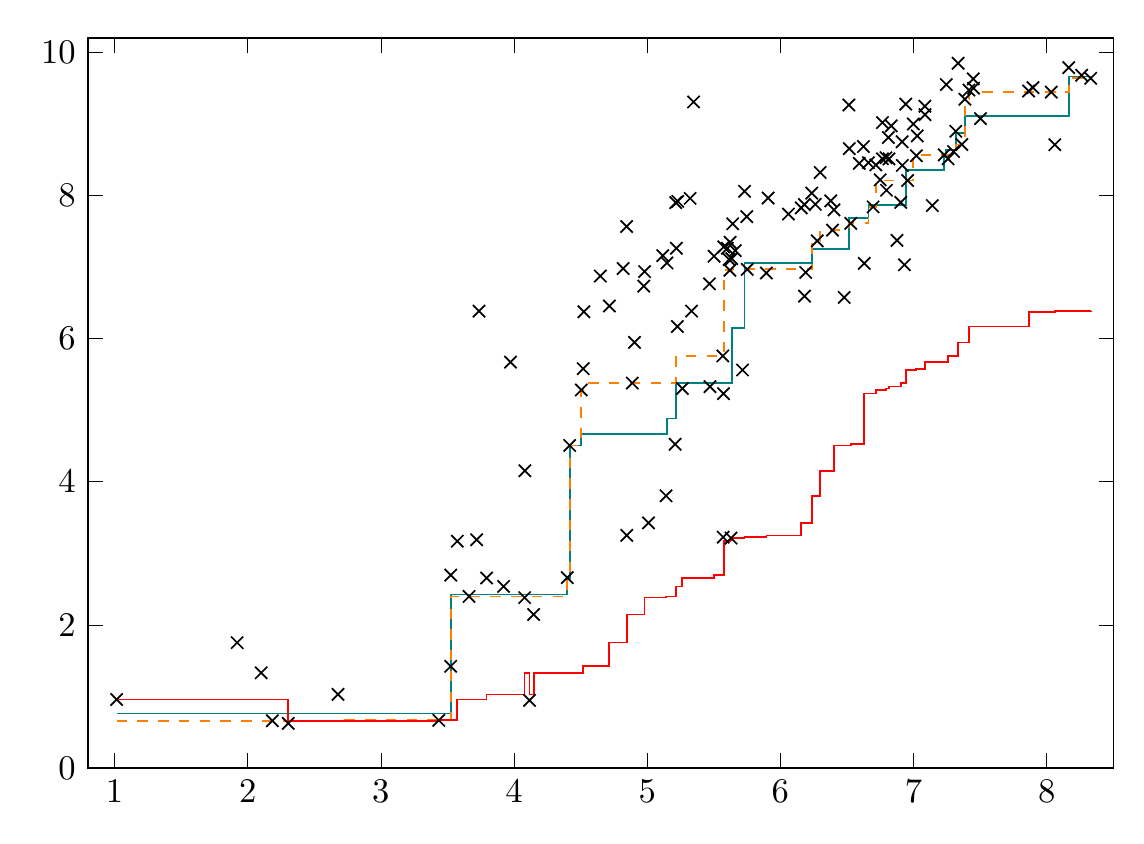} 
		\caption[]%
		{{\small $\tau=0.3$}}    
		\label{fig1:d}
	\end{subfigure}
	\caption[]%
	{\small Illustration of estimated order-$\alpha$, isotonic CQR, and isotonic CER functions. X-axis: ln(C), Y-axis: ln(Q).} 
	\label{fig:fig1}
\end{figure}

The estimated isotonic CQR and isotonic CER functions are step functions enveloping exactly $100\tau\%$ of the observations for each quantile $\tau$. In contrast, the estimated order-$\alpha$ frontier does not necessarily envelope $100\tau\%$ of the observations, but rather less than $100\tau\%$ of the observations especially when the quantile $\tau$ gets smaller such as $\tau=0.3$ (see Figure~\ref{fig1:d}). This observation suggests that the order-$\alpha$ estimator cannot guarantee the quantile property, especially for the low quantile estimation. This is not surprising because order-$\alpha$ is geared towards estimating high quantiles but deteriorates when $\tau$ decreases. Further, the order-$\alpha$ estimator does not even satisfy monotonicity, which is its only assumed shape constraint. The violation of monotonicity occurs in all cases---the estimated order-$\alpha$ frontier (red line) is not strictly increasing but can also decrease, as shown in Figure~\ref{fig:fig1} (see also Figures 2 and 3 in \citename{Daouia2007}, \citeyear*{Daouia2007}). 

\begin{figure}[H]
	\centering
	\begin{subfigure}[b]{0.495\textwidth}
		\centering
		\includegraphics[width=1\textwidth]{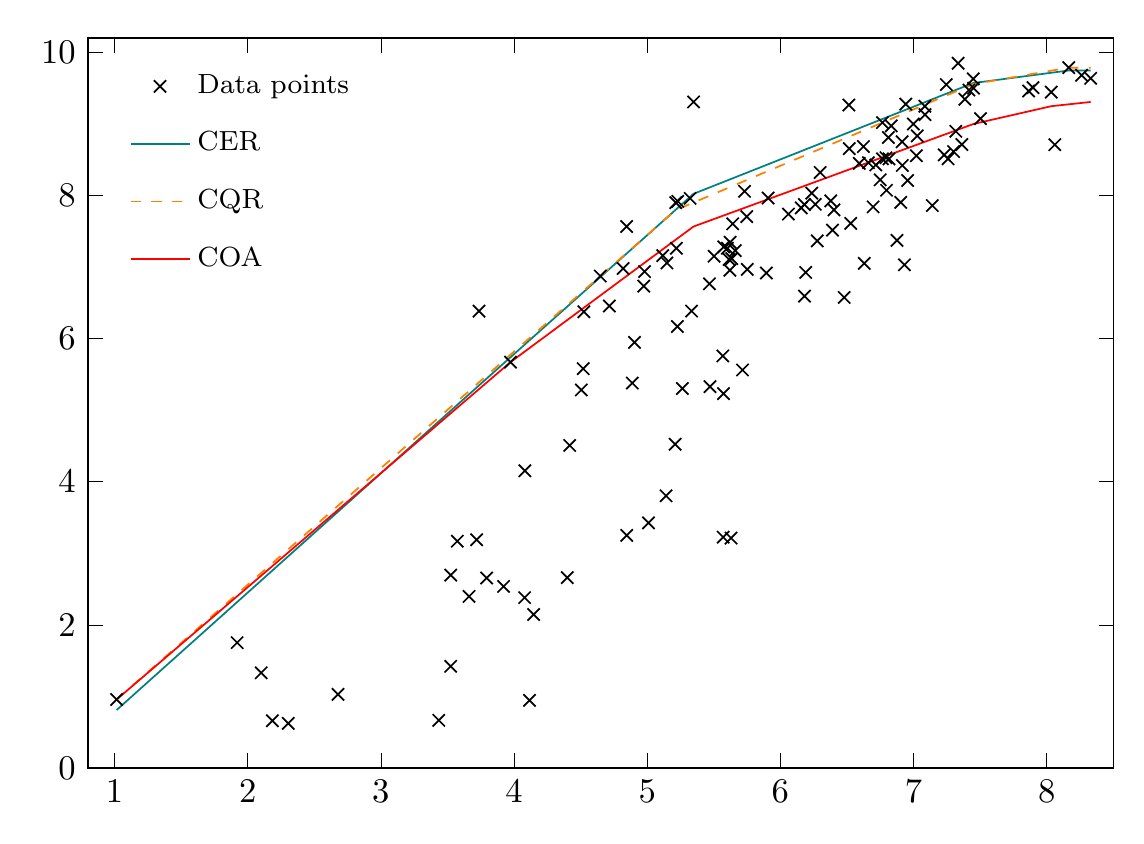} 
		\caption[Network2]%
		{{\small $\tau=0.9$}}    
		\label{fig2:a}
	\end{subfigure}
	\begin{subfigure}[b]{0.495\textwidth}  
		\centering 
		\includegraphics[width=1\textwidth]{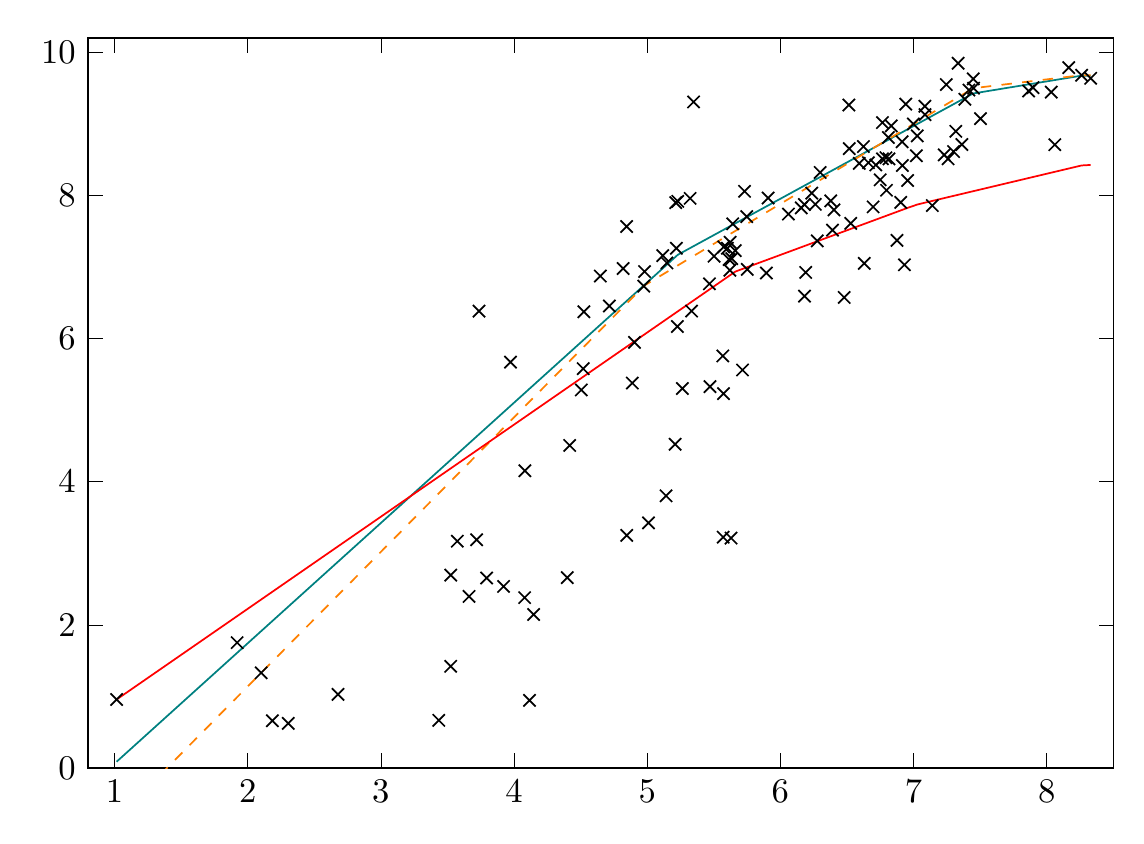} 
		\caption[]%
		{{\small $\tau=0.7$}}    
		\label{fig2:b}
	\end{subfigure}
	\begin{subfigure}[b]{0.495\textwidth}   
		\centering 
		\includegraphics[width=1\textwidth]{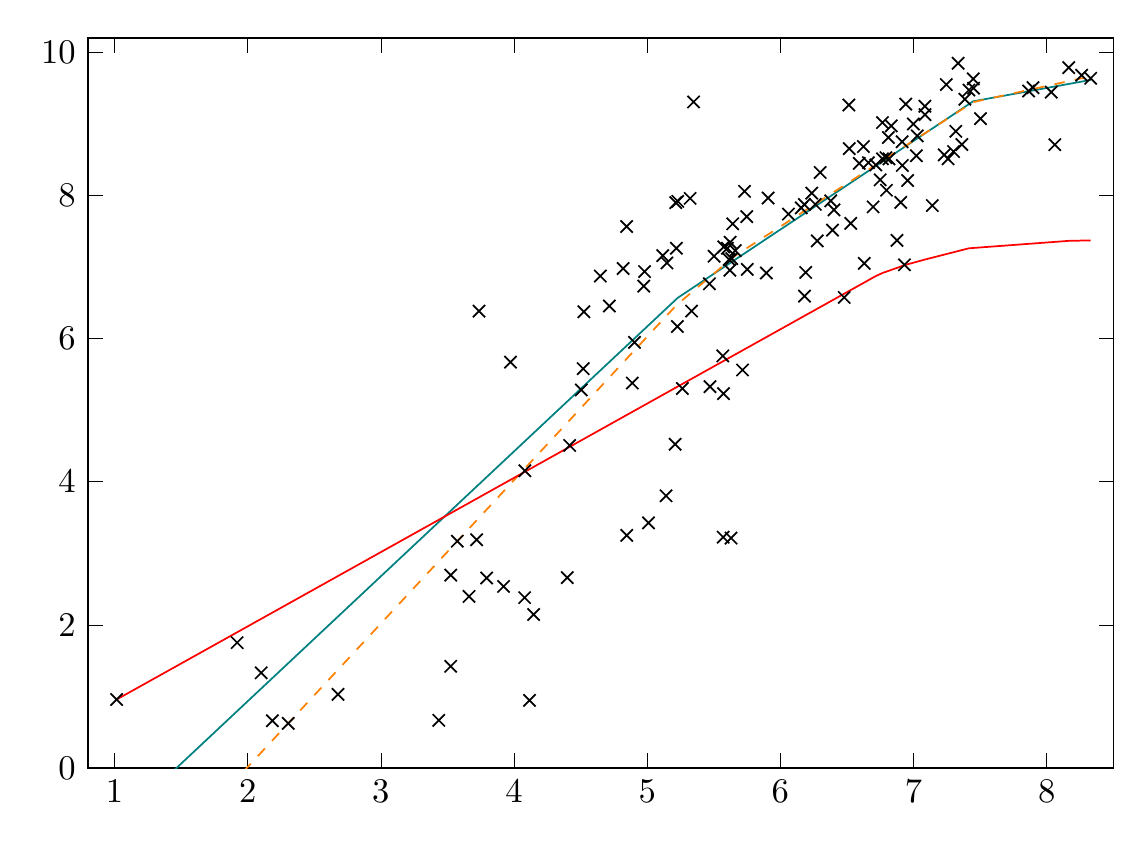} 
		\caption[]%
		{{\small $\tau=0.5$}}    
		\label{fig2:c}
	\end{subfigure}
	\begin{subfigure}[b]{0.495\textwidth}   
		\centering 
		\includegraphics[width=1\textwidth]{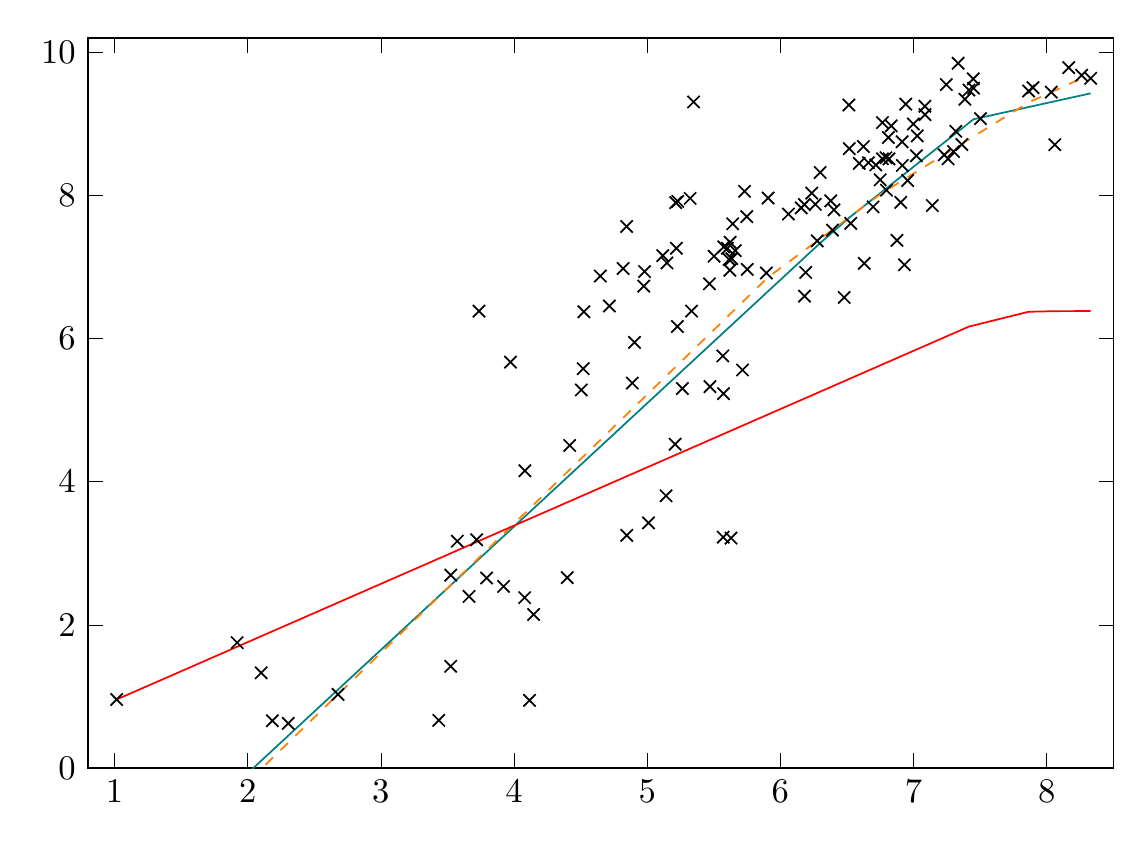} 
		\caption[]%
		{{\small $\tau=0.3$}}    
		\label{fig2:d}
	\end{subfigure}
	\caption[]%
	{\small Illustration of estimated convexified order-$\alpha$, CQR, and CER functions. X-axis: ln(C), Y-axis: ln(Q).} 
	\label{fig:fig2}
\end{figure}

Figure~\ref{fig:fig2} illustrates the estimated convex quantile and expectile functions and the estimated convexified order-$\alpha$ frontier. All three estimators yield a concave piecewise linear curve which can be useful in applications where shadow pricing of non-market inputs and/or outputs is the main object of interest. 

Figures~\ref{fig1:a} and \ref{fig2:a} confirm that as $\tau$ approaches 100\% the estimated order-$\alpha$ frontiers converge to the true unknown full frontier and envelop the observations as much as possible. A large subset of data is thus enough to efficiently estimate the partial frontier. However, the performance of partial frontiers at low quantiles is relatively poor, and quickly deteriorates when $\tau$ gets closer to zero.

Furthermore, indirect estimation of quantiles using expectiles can be a good alternative to estimate monotonic concave quantile functions and even monotonic step quantile functions. For each quantile $\tau$, the indirectly estimated quantile function using expectile regression (teal line) is quite close to the directly estimated quantile function (orange dashed line) (see Figures~\ref{fig:fig1} and \ref{fig:fig2}). 

%

\section{Monte Carlo study}\label{sec:mc}

Having empirically illustrated the estimated quantile functions, we proceed to investigate the finite sample performance of the nonparametric quantile-like estimators through Monte Carlo simulations. The main objective of our simulations is to examine whether the partial frontier estimators can be interpreted as a quantile estimator.

\subsection{Setup}

We generate data according to the following additive Cobb–Douglas production function with $d$ inputs and one output (cf. \citename{Lee2013}, \citeyear*{Lee2013}; \citename{Yagi2018}, \citeyear*{Yagi2018}),
$$y_i = \prod_{d=1}^{D}\bx^{\frac{0.8}{d}}_{d,i} + \varepsilon_i,$$
where the input variables $\bx_i \in \real^{n \times d}$ are randomly and independently drawn from $U[1, 10]$ and the error term $\varepsilon_i$ has three specifications: $\varepsilon_i=v_i$, $\varepsilon_i=-u_i$, and $\varepsilon_i=v_i - u_i$, where $v_i$ and $u_i$ are generated independently from $N(0, \sigma_v^2)$ and $N^+(0, \sigma_u^2)$, respectively. $\sigma_v^2$ and $\sigma_u^2$ are determined once we set signal to noise ratio (SNR) $\lambda$ and variance $\sigma^2$, where $\lambda = \sigma_u / \sigma_v$ and $\sigma^2 = \sigma_u^2 + \sigma_v^2$. Following \citeasnoun{Aigner1977}, ($\sigma^2$, $\lambda$) = (1.88, 1.66), (1.63, 1.24), and (1.35, 0.83) are selected\footnote
{
This corresponds to $\sigma_u$ = 1.174, 0.994, 0.742 and $\sigma_v$ =0.708, 0.801, 0.894, respectively.
} which allow for investigating whether those quantile-like estimators are robust to a wide range of SNR values.

To assess the finite sample performance of the quantile-like estimators, we utilize the standard mean squared error (MSE) and bias statistics to evaluate how the estimated quantile function deviates from the true conditional quantile function. The MSE and bias statistics can be defined as
\[
\text{MSE}=\frac{1}{n} \sum\limits_{i}^{n}\left({{\widehat{Q}}_{y_i}}(\tau\,|\,\bx_i)-{Q_{y_i}}(\tau\,|\,\bx_i)\right)^2,
\]\[
\text{bias}=\frac{1}{n} \sum\limits_{i}^{n}{\left({\widehat Q_{y_i}}(\tau\,|\,\bx_i)-{Q_{y_i}}(\tau\,|\,\bx_i)\right)},
\]
where $\hat{Q}_{y_i}$ denotes the estimated conditional quantile function and $Q_{y_i}$ represents the true conditional quantile function; the latter can be estimated based on the known inverse cumulative distribution function of the error term $\varepsilon_i$, i.e., $F_{\varepsilon_i}^{-1}(\tau)$. The MSE is always greater than or equal to zero, with zero indicating perfect precision; while the bias can be negative, positive, or zero, suggesting whether the estimated conditional quantile function $\hat{Q}_{y_i}$ systematically underestimates ($\text{bias} < 0$), overestimates ($\text{bias} > 0$), or provides an unbiased estimate of ($\text{bias} = 0$) the true conditional quantile function.

In all experiments that follow, we resort to Julia/JuMP (\citename{Dunning2017}, \citeyear*{Dunning2017}) to solve the CQR/CER and isotonic CQR/CER estimators with the commercial off-the-shelf solver MOSEK (9.3).\footnote{
Alternatively, the estimation of CQR/CER and isotonic CQR/CER can be implemented in Python using the pyStoNED package (\citename{Dai2021b}, \citeyear*{Dai2021b}). 
}
The original and convexified order-$\alpha$ estimators are computed using the R packages ``frontiles'' (\citename{Daouia2020}, \citeyear*{Daouia2020}) and ``Benchmarking'' (\citename{Bogetoft2010}, \citeyear*{Bogetoft2010}). All experiments are run on Aalto University's high-performance computing cluster Triton with Xeon @2.8 GHz processors, one CPU, and 3 GB of RAM per task.

\subsection{Experiment with monotonic estimators}\label{sec:exp1}

In the first group of experiments, we explore whether the nonconvex quantile estimator (i.e., isotonic CQR/CER) has better finite sample performance than the nonconvex partial frontier estimator (i.e., order-$\alpha$) in estimating the quantile production functions. We consider 225 scenarios with different numbers of observations (50, 100, 200, 500, and 1000), input dimensions (1, 2, and 3), SNRs (1.66, 1.24, and 0.83), and quantiles (0.1, 0.3, 0.5, 0.7, and 0.9). Each scenario is replicated 1000 times to calculate the MSE and bias statistics. For the sake of comparison, the expectiles $\tilde{\tau}$ are transformed into their corresponding quantiles $\tau$ based on the empirical inverse quantile function of the error term $\varepsilon_i$.

\begin{table}[H]
	\centering
	\caption{Performance in estimating monotonic quantile function $Q_y$ with $\sigma^2=1.88$ and $\tau=0.9$. ICQR = Isotonic CQR, ICER = Isotonic CER.}
    \begin{tabular}{rrrrrrrrr}
	\toprule
	\multicolumn{1}{c}{\multirow{2}[4]{*}{$d$}} & \multicolumn{1}{c}{\multirow{2}[4]{*}{$n$}} & \multicolumn{3}{c}{MSE} &       & \multicolumn{3}{c}{Bias} \\
	\cmidrule{3-5}\cmidrule{7-9}          &       & ICQR  & ICER  & Order-$\alpha$ &       & ICQR  & ICER  & Order-$\alpha$ \\
	\midrule
	1     & 50    & 0.368 & 0.406 & 1.470 &       & -0.284 & -0.385 & -0.969 \\
	& 100   & 0.215 & 0.231 & 1.479 &       & -0.166 & -0.252 & -1.000 \\
	& 200   & 0.132 & 0.135 & 1.419 &       & -0.097 & -0.156 & -1.003 \\
	& 500   & 0.069 & 0.067 & 1.409 &       & -0.051 & -0.086 & -1.012 \\
	& 1000  & 0.042 & 0.039 & 1.404 &       & -0.031 & -0.054 & -1.018 \\
	&   \vspace{-0.7em}\\ 
	2     & 50    & 0.933 & 0.989 & 1.777 &       & -0.671 & -0.731 & -1.076 \\
	& 100   & 0.639 & 0.692 & 1.784 &       & -0.522 & -0.591 & -1.121 \\
	& 200   & 0.416 & 0.454 & 1.742 &       & -0.387 & -0.454 & -1.131 \\
	& 500   & 0.236 & 0.255 & 1.712 &       & -0.261 & -0.313 & -1.146 \\
	& 1000  & 0.150 & 0.160 & 1.692 &       & -0.186 & -0.231 & -1.150 \\
	&   \vspace{-0.7em}\\ 
	3     & 50    & 1.479 & 1.519 & 1.959 &       & -0.912 & -0.944 & -1.115 \\
	& 100   & 1.152 & 1.197 & 1.901 &       & -0.787 & -0.827 & -1.129 \\
	& 200   & 0.875 & 0.920 & 1.882 &       & -0.668 & -0.714 & -1.158 \\
	& 500   & 0.572 & 0.602 & 1.849 &       & -0.514 & -0.558 & -1.181 \\
	& 1000  & 0.405 & 0.425 & 1.820 &       & -0.415 & -0.455 & -1.191 \\
	\bottomrule
	\end{tabular}%
	\label{tab:tab2}%
\end{table}%

Table~\ref{tab:tab2} reports the effect of sample size on the performance of each estimator in the case of $\tau = 0.9$, a commonly used parameter value in the robust frontier estimation. The results show that the finite sample performance of isotonic CQR and isotonic CER is superior to that of order-$\alpha$ in terms of both MSE and bias statistics. Further, the performance of each estimator improves with a larger sample size $n$, as expected. Specifically, the MSE and bias statistics of isotonic CQR and isotonic CER estimators get closer to zero as $n$ increases, which suggests that both estimators are consistent. The MSE of order-$\alpha$ also generally falls as the sample size increases, whereas the bias does not diminish as the sample size increases due to losing the $\sqrt{n}$-consistency (\citename{Aragon2005}, \citeyear*{Aragon2005}). 

Next, consider the choice of quantiles $\tau$. Figure~\ref{fig:ex1} depicts the MSE results in estimating the quantile functions for different input dimensions and SNR specifications, while keeping the sample size fixed at $n=1000$. In all scenarios considered, isotonic CQR and isotonic CER have far smaller MSE values than order-$\alpha$. However, the difference in terms of MSE between isotonic CQR and isotonic CER is quite small. Another interesting observation is that when the quantile $\tau$ becomes smaller, the MSE of order-$\alpha$ sees a systematic increasing trend.
\begin{figure}[H]
\begin{center}
\includegraphics[width=1\textwidth]{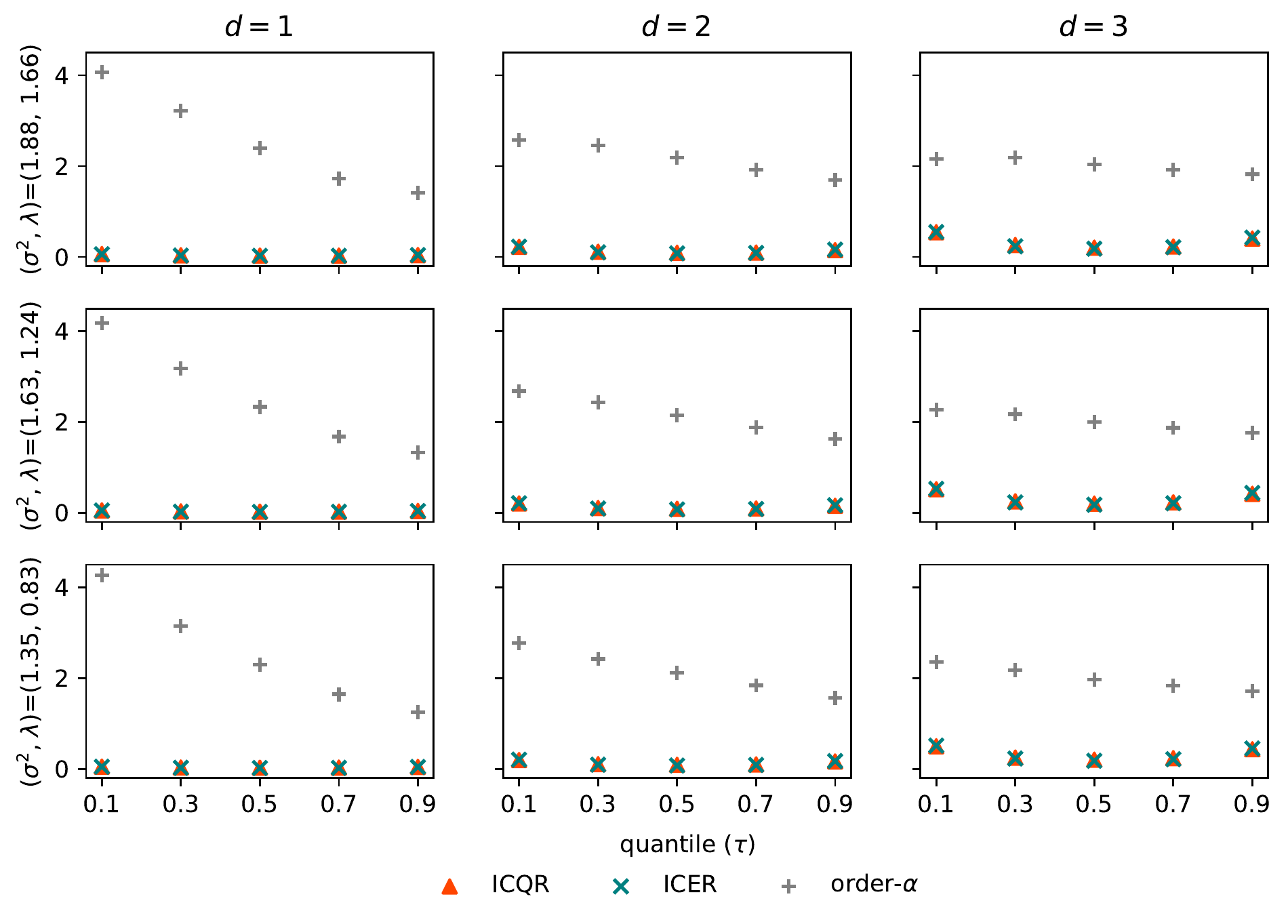}
\end{center}
\vspace{-1em}
\caption{MSE results of order-$\alpha$, isotonic CQR, and isotonic CER with $n=1000$.}
\label{fig:ex1}
\end{figure}

We note that the MSE of each estimator generally increases as more input variables are introduced. This is because a larger dimensionality increases the data sparsity, which degrades the performance of each estimator, \textit{ceteris paribus}. For example, when $\tau = 0.9$ and $\sigma^2 = 1.88$, order-$\alpha$'s MSE increases from 1.40 in the one-input case to 1.69 in the two-input case to 1.82 in the three-input case, and isotonic CQR's and isotonic CER's MSE values rise from 0.04 to 0.15 to 0.40 and from 0.04 to 0.16 to 0.42, respectively. A similar curse of dimensionality also exists in the DEA simulation studies, where the performance of DEA deteriorates when the number of inputs increases, \textit{ceteris paribus} (see, e.g., \citename{Pedraja-Chaparro1999a}, \citeyear*{Pedraja-Chaparro1999a}; \citename{Cordero2015}, \citeyear*{Cordero2015}). 
\begin{figure}[H]
\begin{center}
\includegraphics[width=1\textwidth]{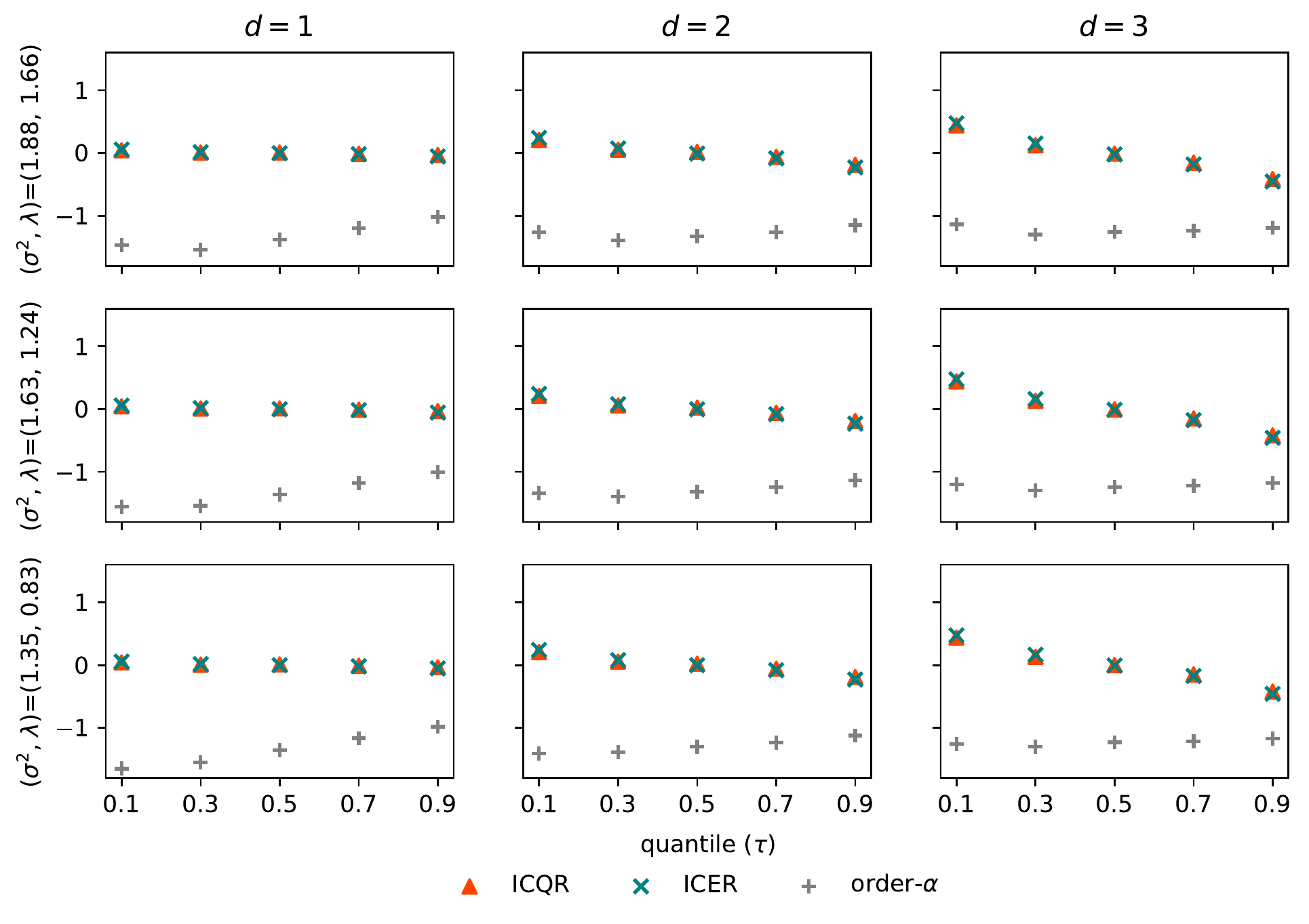}
\end{center}
\vspace{-1em}
\caption{Bias results of order-$\alpha$, isotonic CQR, and isotonic CER with $n = 1000$.}
\label{fig:ex2}
\end{figure}  

Figure~\ref{fig:ex2} displays the bias results. The isotonic CQR and isotonic CER estimators yield both positive and negative biases. The bias gets greater (in terms of the absolute value) when $\tau$ deviates from 0.5: it becomes a larger positive value when $\tau$ decreases from 0.5 and, on the opposite, a smaller negative value when $\tau$ increases from 0.5. By contrast, the order-$\alpha$ estimator yields only negative biases. Since the order-$\alpha$ frontier converges to the FDH full frontier in a finite sample when $\tau \xrightarrow{} 1$, the observed negative bias of order-$\alpha$ for each quantile $\tau$ is due to the small sample bias, similar to FDH. Moreover, the bias of order-$\alpha$ becomes larger as $\tau$ decreases because the effective sample size gets smaller. 

Furthermore, we obtain similar results about the MSE and bias statistics and the sample size effect in additional experiments where the composite error term $\varepsilon_i$ contains either inefficiency ($\varepsilon_i = - u_i$) or noise ($\varepsilon_i=v_i$) (see Appendix~\ref{app:err}). We also investigate the estimators' performance in the presence of functional form misspecification and find that isotonic CQR and isotonic CER outperforms order-$\alpha$ in terms of both MSE and bias statistics (see Appendix~\ref{app:miss}). To examine the robustness of each estimator, we consider additional scenarios with outliers. The results suggest that the isotonic CER estimator is superior in all scenarios, and isotonic CER and isotonic CQR are more robust than the order-$\alpha$ estimator due to the fact that order-$\alpha$ does not satisfy the quantile property (see Appendix~\ref{app:out}). 

Another point worth noting is that the order-$\alpha$ estimator is found to perform relatively poorly at low quantiles. This contradicts the fact that a quantile estimator should perform virtually equally well at all quantiles, as suggested by the quantile property (see parts i) and ii) of Theorem \ref{the:the1}). Thus, we further investigate the frequency that the monotonic estimators would violate the quantile property in 1000 replications.
\begin{table}[H]
	\centering
	\caption{Frequency of quantile property violations for order-$\alpha$ in 1000 replications.}
     \small{
	\begin{tabular}{rrrrrrrrrrrr}
		\toprule
		\multicolumn{1}{c}{\multirow{2}[4]{*}{$n$}} & \multicolumn{1}{c}{\multirow{2}[4]{*}{$d$}} & \multicolumn{1}{c}{\multirow{2}[4]{*}{($\sigma^2$,$\lambda$)}} & \multicolumn{3}{c}{$\tau$} &       & \multicolumn{1}{c}{\multirow{2}[4]{*}{$n$}} & \multicolumn{1}{c}{\multirow{2}[4]{*}{$d$}} & \multicolumn{1}{c}{\multirow{2}[4]{*}{($\sigma^2$,$\lambda$)}} & \multicolumn{2}{c}{$\tau$} \\
		\cmidrule{4-6}\cmidrule{11-12}          &       &       & 0.1   & 0.3   & 0.5   &       &       &       &       & 0.1   & 0.3 \\
    \midrule
    50    & 1     & (1.88, 1.66) & 95.8 \% & 0.3 \% &       &       & 100   & 1     & (1.88, 1.66) & 99.6 \% &  \\
          &       & (1.63, 1.24) & 90.3 \% & 0.3 \% &       &       &       &       & (1.63, 1.24) & 97.7 \% &  \\
          &       & (1.35, 0.83) & 75.1 \% & 0.4 \% &       &       &       &       & (1.35, 0.83) & 86.7 \% &  \\
          & 2     & (1.88, 1.66) & 89.6 \% & 1.7 \% &       &       &       & 2     & (1.88, 1.66) & 88.8 \% &  \\
          &       & (1.63, 1.24) & 83.8 \% & 1.6 \% &       &       &       &       & (1.63, 1.24) & 76.8 \% &  \\
          &       & (1.35, 0.83) & 79.7 \% & 1.9 \% &       &       &       &       & (1.35, 0.83) & 66.1 \% & 0.1 \% \\
          & 3     & (1.88, 1.66) & 98.9 \% & 7.6 \% & 0.2 \% &      &       & 3     & (1.88, 1.66) & 98.1 \% & 0.5 \% \\
          &       & (1.63, 1.24) & 98.8 \% & 8.2 \% & 0.2 \% &      &       &       & (1.63, 1.24) & 97.3 \% & 0.7 \% \\
          &       & (1.35, 0.83) & 98.7 \% & 7.1 \% & 0.1 \% &      &       &       & (1.35, 0.83) & 96.7 \% & 0.5 \% \\
          &   \vspace{-0.7em}\\ 
    500   & 1     & (1.88, 1.66) & 100.0 \% &      &       &       & 1000  & 1     & (1.88, 1.66) & 100.0 \% &  \\
          &       & (1.63, 1.24) & 100.0 \% &      &       &       &       &       & (1.63, 1.24) & 100.0 \% &  \\
          &       & (1.35, 0.83) & 98.7 \% &       &       &       &       &       & (1.35, 0.83) & 100.0 \% &  \\
          & 2     & (1.88, 1.66) & 89.3 \% &       &       &       &       & 2     & (1.88, 1.66) & 90.1 \% &  \\
          &       & (1.63, 1.24) & 32.6 \% &       &       &       &       &       & (1.63, 1.24) & 10.8 \% &  \\
          &       & (1.35, 0.83) & 10.0 \% &       &       &       &       &       & (1.35, 0.83) & 0.9 \% &  \\
          & 3     & (1.88, 1.66) & 74.5 \% &       &       &       &       & 3     & (1.88, 1.66) & 37.9 \% &  \\
          &       & (1.63, 1.24) & 60.2 \% &       &       &       &       &       & (1.63, 1.24) & 17.1 \% &  \\
          &       & (1.35, 0.83) & 50.7 \% &       &       &       &       &       & (1.35, 0.83) & 11.3 \% &  \\
    \bottomrule
	\end{tabular}%
 }
  \begin{tablenotes}
        \setlength\labelsep{0pt}
        \footnotesize
        \item \textit{Note:} The blanks in the columns of different quantiles denote zero violations.
        \end{tablenotes}
	 \label{tab:tab3}%
\end{table}%

Our simulations confirm that both isotonic CQR and isotonic CER exactly satisfy the quantile property with the violation rates being zero. In contrast, the quantile property is systematically violated in order-$\alpha$ at low quantiles, particularly at the 10\% quantile (see Table~\ref{tab:tab3}). The observed violations are due to the fact that order-$\alpha$ relies on the quantiles of an appropriate distribution based on a subset of the sample. However, for high quantiles (i.e., $\tau > 0.5$), the violation rates in order-$\alpha$ are also equal to zero, suggesting that order-$\alpha$ can satisfy the quantile property for large $\tau$. This is consistent with the findings from the MSE and bias comparisons. In conclusion, the Monte Carlo simulations presented in this sub-section demonstrate that the true quantile estimators perform notably better than the partial frontiers in the nonconvex case.

\subsection{Experiment with monotonic and convex estimators}\label{sec:exp2}

We next conduct the second group of experiments to compare the performance of the convex estimators (i.e., CQR, CER, and convexified order-$\alpha$) using the same scenarios as in Section \ref{sec:exp1}. Table~\ref{tab:tab4} presents the effects of sample size and dimensionality on the MSE and bias statistics for $\tau = 0.9$. Figures~\ref{fig:ex3} and \ref{fig:ex4} display the MSE and bias statistics of the convexified order-$\alpha$, CQR, and CER estimators as we alternate the values of $\tau$ and SNR, while keeping the sample size constant at $n=1000$. 

The simulation results reported in Table~\ref{tab:tab4} suggest that both CQR and CER estimators exhibit superior performance compared to the convexified order-$\alpha$ estimator both in terms of MSE and bias. Further, the MSE and bias of CQR and CER converge towards zero as the sample size $n$ increases, while this is not the case for the convexified order-$\alpha$ estimator when the dimensionality $d = 1, 2$. 

Comparing Figures~\ref{fig:ex1} and \ref{fig:ex3}, we notice that the MSE statistic for each estimator decreases to a great extent once imposing the concavity constraint, especially for order-$\alpha$. For instance, in the one-input case with $\sigma=1.88$, the average MSE of convexified order-$\alpha$ for the five estimated quantiles decreases by more than 160\% compared to its original counterpart. This finding confirms that the power of the CQR, CER, and convexified order-$\alpha$ estimators derives from their global shape constraints, including monotonicity and convexity/concavity (\citename{Kuosmanen2020}, \citeyear*{Kuosmanen2020}).
\begin{table}[H]
	\centering
	\caption{Performance in estimating monotonic and concave quantile function $Q_y$ with $\sigma^2=1.88$ and $\tau=0.9$.CQR = Convex quantile regression, CER = Convex expectile regression, COA = Convexified order-$\alpha$.}
	\begin{tabular}{rrrrrrrrr}
		\toprule
		\multicolumn{1}{c}{\multirow{2}[4]{*}{$d$}} & \multicolumn{1}{c}{\multirow{2}[4]{*}{$n$}} & \multicolumn{3}{c}{MSE} &       & \multicolumn{3}{c}{Bias} \\
		\cmidrule{3-5}\cmidrule{7-9}          &       & \multicolumn{1}{l}{CQR} & \multicolumn{1}{l}{CER} & \multicolumn{1}{l}{COA} &       & \multicolumn{1}{l}{CQR} & \multicolumn{1}{l}{CER} & \multicolumn{1}{l}{COA} \\
		\midrule
		1     & 50    & 0.170 & 0.169 & 0.635 &       & -0.056 & -0.110 & -0.497 \\
		& 100   & 0.094 & 0.088 & 0.704 &       & -0.023 & -0.056 & -0.574 \\
		& 200   & 0.050 & 0.050 & 0.757 &       & -0.009 & -0.023 & -0.636 \\
		& 500   & 0.022 & 0.021 & 0.831 &       & -0.005 & -0.010 & -0.694 \\
		& 1000  & 0.031 & 0.011 & 0.880 &       & -0.005 & -0.005 & -0.738 \\
		&   \vspace{-0.7em}\\ 
		2     & 50    & 0.377 & 0.395 & 0.838 &       & -0.194 & -0.308 & -0.574 \\
		& 100   & 0.231 & 0.239 & 0.854 &       & -0.106 & -0.186 & -0.601 \\
		& 200   & 0.133 & 0.137 & 0.880 &       & -0.056 & -0.102 & -0.621 \\
		& 500   & 0.067 & 0.069 & 0.934 &       & -0.027 & -0.048 & -0.654 \\
		& 1000  & 0.039 & 0.039 & 0.964 &       & -0.013 & -0.026 & -0.668 \\
		&   \vspace{-0.7em}\\ 
		3     & 50    & 0.632 & 0.671 & 0.857 &       & -0.406 & -0.502 & -0.574 \\
		& 100   & 0.412 & 0.438 & 0.741 &       & -0.249 & -0.344 & -0.485 \\
		& 200   & 0.263 & 0.280 & 0.709 &       & -0.152 & -0.229 & -0.440 \\
		& 500   & 0.141 & 0.150 & 0.722 &       & -0.082 & -0.126 & -0.427 \\
		& 1000  & 0.087 & 0.091 & 0.732 &       & -0.046 & -0.078 & -0.419 \\
		\bottomrule
	\end{tabular}%
	\label{tab:tab4}%
\end{table}%

While the performance of order-$\alpha$ increases after imposing the concavity constraint, the CQR and CER estimators continue to outperform convexified order-$\alpha$ in all cases considered. However, the relative MSE ratio between convexified order-$\alpha$ and CQR (or CER) decreases as the input dimensionality or the quantile $\tau$ increases. Regarding the effect of different SNRs, the smaller the value of $\lambda$, the higher the difference in MSE between the quantile and partial frontier estimators. However, the difference in MSE among the three SNRs is close to zero when the quantile $\tau$ approaches 1. 
\begin{figure}[H]
    \centering
    \includegraphics[width=1\textwidth]{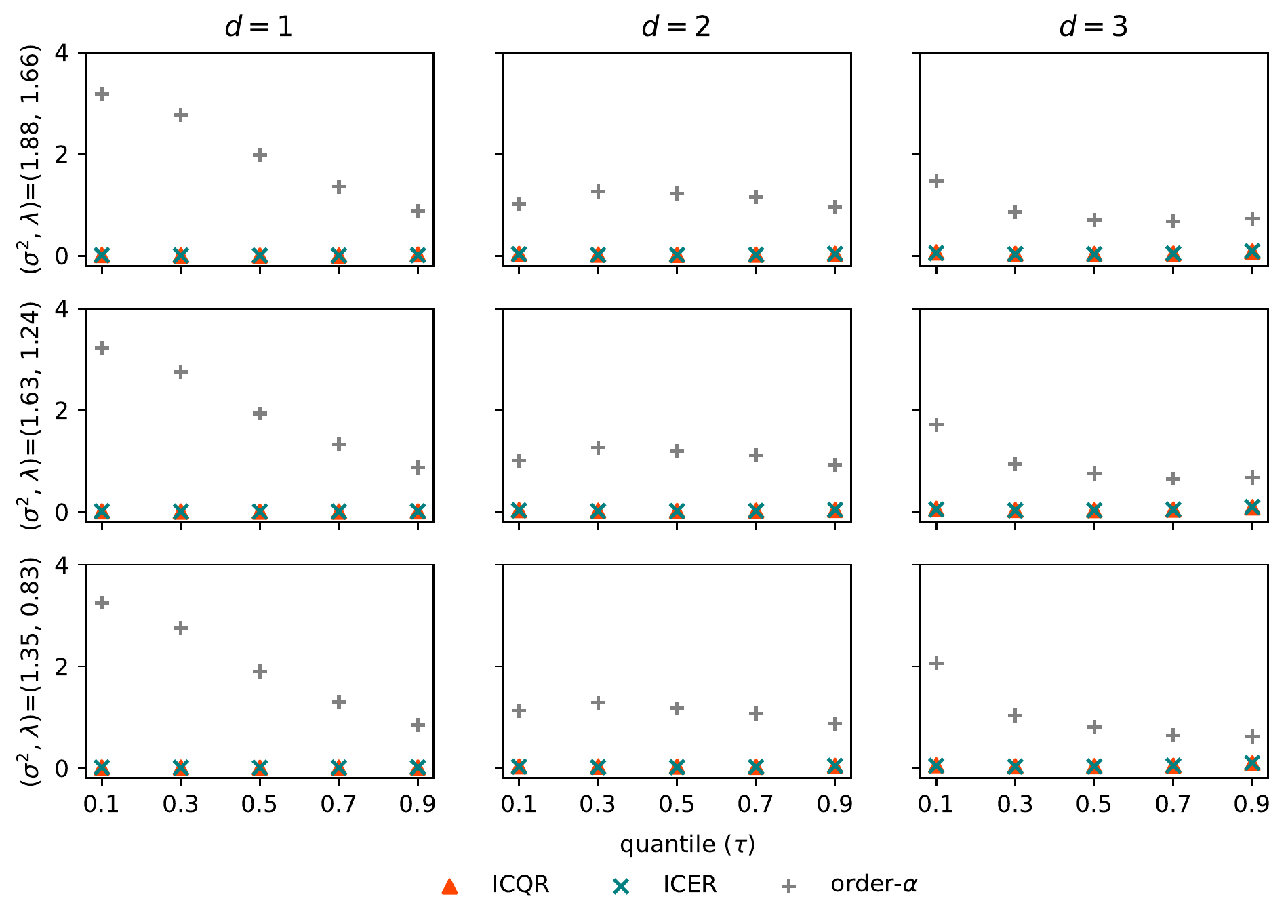}
    \caption{MSE results of convexified order-$\alpha$, CQR, and CER with $n=1000$.}
    \label{fig:ex3}
\end{figure}

Recall that the biases of order-$\alpha$ in all considered scenarios are negative, indicating that the estimated partial frontiers systematically underestimate the true quantile functions. After imposing the concavity constraint, for the three-input cases, the convexified order-$\alpha$ estimator does not only underestimate but can also overestimate the true quantile function. Moreover, the absolute bias of convexified order-$\alpha$ is larger than that of CQR/CER in all scenarios. Note that CQR and CER can better fit the true quantile functions with lower MSE and bias values compared to the monotonic estimators in Section \ref{sec:exp1}.
\begin{figure}[H]
    \centering
    \includegraphics[width=1\textwidth]{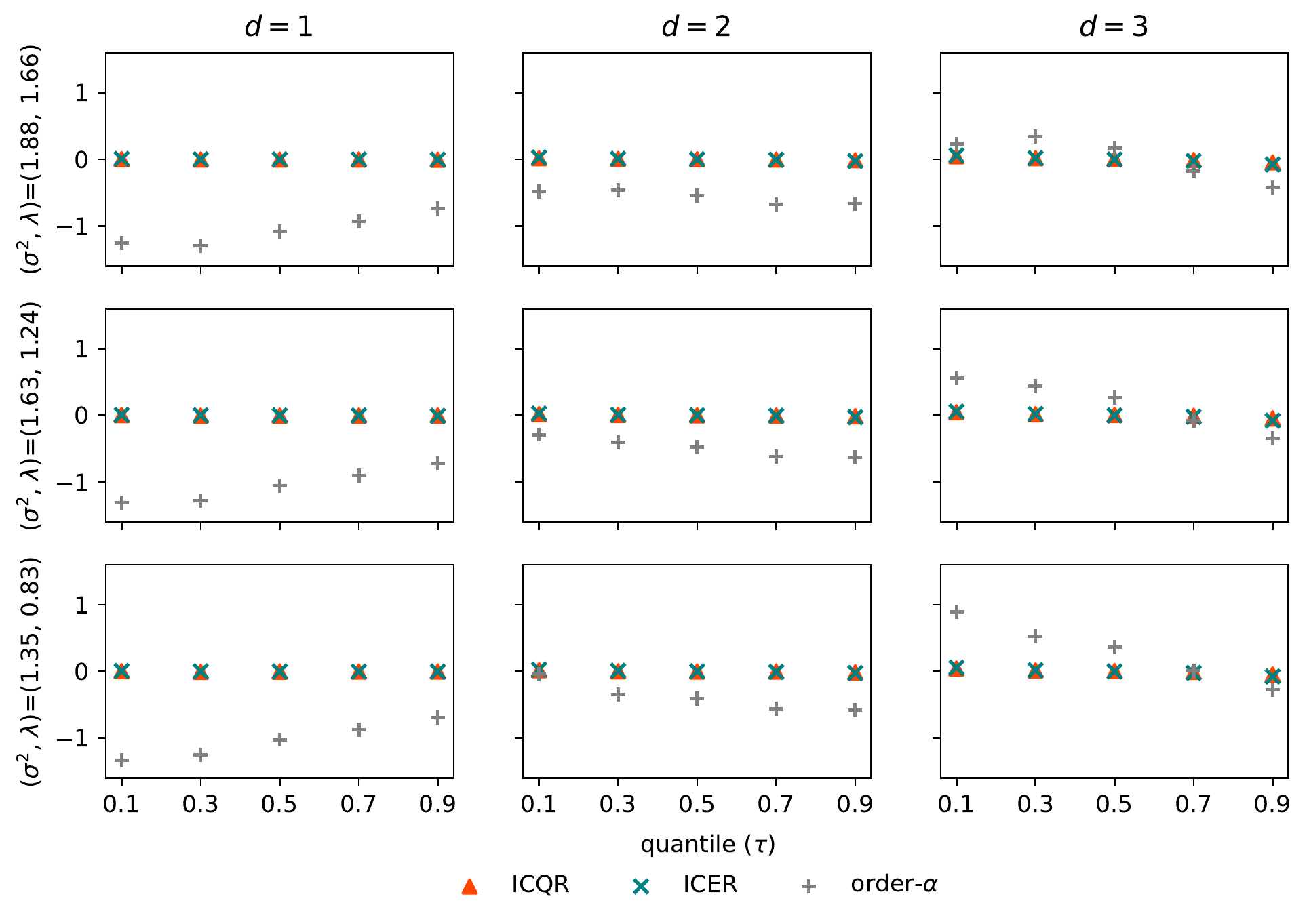}
    \caption{Bias results of convexified order-$\alpha$, CQR, and CER with $n=1000$.}
    \label{fig:ex4}
\end{figure}

The simulation results in Sections~\ref{sec:exp1} and \ref{sec:exp2} reveal that the indirect estimation of quantiles using expectiles improves the performance in most scenarios considered, particularly for the concave quantile functions. Table~\ref{tab:tab5} reports the percentage of simulation rounds where the MSE of the indirect expectile estimators is lower than that of the direct quantile estimators. Compared to isotonic CQR, isotonic CER has smaller MSE values for most quantiles considered except for those extreme quantiles (e.g., the 10\% and 90\% quantiles). Further, when we impose the concavity constraint, the CER estimator outperforms the CQR estimator in a larger proportion of scenarios (e.g., all scenarios at the 10\% and 50\% quantiles). The observation from Table~\ref{tab:tab5} suggests that the indirect estimation of quantiles through expectiles performs better when $\tau$ is close to 0.5, whereas the direct quantile estimation remains competitive when $\tau$ is very small or very large.
\begin{table}[H]
	\centering
	\caption{Percentage of expectile estimation's MSE less than quantile estimation's MSE.}
	\begin{tabular}{rrrrrr}
		\toprule
		\multicolumn{1}{c}{Model specification} & \multicolumn{1}{c}{$\tau$} & \multicolumn{1}{c}{$\varepsilon=v-u$} & \multicolumn{1}{l}{$\varepsilon=v$} & \multicolumn{1}{c}{$\varepsilon=-u$} & \multicolumn{1}{c}{No. scenarios} \\
		\midrule
		\multicolumn{1}{l}{Monotonicity} & \multicolumn{1}{l}{all quantiles} & 65.8 \% & 63.1 \% & 57.3 \% & 225 \\
		& 0.1   & 22.2 \% & 13.3 \% & 26.7 \% & 45 \\
		& 0.5   & 100.0 \% & 100.0 \% & 100.0 \% & 45 \\
		& 0.9   & 13.3 \% & 13.3 \% & 13.3 \% & 45 \\
		&   \vspace{-0.7em}\\ 
		\multicolumn{1}{l}{+ Concavity} & \multicolumn{1}{l}{all quantiles} & 88.0 \% & 88.0 \% & 89.8 \% & 225 \\
		& 0.1   & 100.0 \% & 100.0 \% & 100.0 \% & 45 \\
		& 0.5   & 100.0 \% & 100.0 \% & 100.0 \% & 45 \\
		& 0.9   & 40.0 \% & 40.0 \% & 55.6 \% & 45 \\
		\bottomrule
	\end{tabular}%
	\label{tab:tab5}%
\end{table}%

%

\section{Conclusions}\label{sec:conc}

Partial frontiers are originally developed to increase robustness to outliers and extreme data, which are often referred to as quantiles in the literature. However, whether the partial frontiers are truly quantile functions remains an open question.

To address this issue, we have extended CQR and CER as alternatives of order-$\alpha$ to estimate monotonic quantile functions and, on the other hand, developed convexified order-$\alpha$ to facilitate comparisons with CQR and CER. These quantile-like estimators are applied to an empirical dataset of U.S. electric power plants to illustrate and visualize what the estimated partial frontiers and quantiles look like. Further, the finite sample performance of these estimators is numerically compared by Monte Carlo simulations.

The empirical application demonstrates that the estimated isotonic CQR and CER functions are step functions enveloping exactly $100\tau\%$ of the observations for each quantile $\tau$. In contrast, the estimated order-$\alpha$ frontier does not necessarily envelope $100\tau\%$ of the observations, but rather less than $100\tau\%$ of the observations especially when the quantile $\tau$ gets smaller. The order-$\alpha$ estimator does not satisfy monotonicity in all considered cases. It is geared towards estimating high quantiles but deteriorates when $\tau$ decreases.

The Monte Carlo simulations have demonstrated that partial frontiers perform relatively poorly in terms of both MSE and bias compared to quantile estimators, especially for low quantiles. The partial frontier estimators fit a frontier to a subset of data, which results in inefficient utilization of the information in the full sample and violates the quantile property of Theorem \ref{the:the1} at most quantiles. Furthermore, correctly imposed convexity improves the performance in both partial frontier and quantile function approaches, but clearly, quantile estimators are better for estimating quantiles than partial frontiers whether convexity is imposed or not. The simulation evidence also shows that the indirect quantile estimators through expectile regression exhibit better finite sample performance over the direct quantile estimators in most cases.

The findings drawn from this paper can provide insights into shaped constrained quantile regression. However, there are several fascinating avenues for future research. We have deliberately kept away from statistical inferences, and that further work in this direction, e.g., how to apply bootstrapping to CQR, CER and their nonconvex counterparts, would be needed. Another avenue for future research is to extend the current research to multi-input and multi-output settings.

It is worth noting that the partial frontiers tend to be used in efficiency measurement and benchmarking, whereas the quantile estimators have become increasingly used in the context of shadow pricing. We hope that the systematic review and performance comparison in a controlled environment of Monte Carlo simulations could help to facilitate further exchange and interaction between these two separate streams of literature.

%

\section*{Supplementary materials}

The supplementary materials contain detailed proof of Theorem 3 in Section \ref{sec:icqr} and additional simulation experiments results in Section \ref{sec:mc}.

%

\section*{Acknowledgments}\label{sec:ack}

We acknowledge the computational resources provided by the Aalto Science-IT project. Sheng Dai gratefully acknowledges financial support from the Foundation for Economic Education (Liikesivistysrahasto) [grants no. 180019, 190073, 21007] and the HSE Support Foundation [grant no. 11--2290]. Xun Zhou gratefully acknowledges financial support from the Finnish Cultural Foundation [grant no. 00201201].

%

\bibliography{References}

%

\newpage
\section*{Appendix}\label{sec:app}
\renewcommand{\thesubsection}{\Alph{subsection}}
\setcounter{table}{0}
\setcounter{figure}{0}
\setcounter{equation}{0}
\setcounter{theorem}{0}
\renewcommand{\theequation}{A\arabic{equation}} 
\renewcommand{\thetable}{B\arabic{table}} 

\subsection{Proof of Theorem 3}\label{app:proof3}

We can rewrite isotonic CQR \eqref{eq:icqr} and isotonic CER \eqref{eq:icer} as the respective equivalent problems according to the quantile and expectile regression definitions (\citename{Koenker1978}, \citeyear*{Koenker1978}; \citename{Newey1987}, \citeyear*{Newey1987}). Specifically, isotonic CQR \eqref{eq:icqr} can be reformulated as
\begin{alignat}{2}
	\underset{\alpha,\mathbf{\bbeta}}{\mathop{\min }}&\,\tau \sum\limits_{i=1}^{n}{\rho_\tau(y_i - \alpha_i - \bbeta_{i}^{'}{{\bx}_{i}})} \label{a1}\\
	\mbox{\textit{s.t.}}\quad
	& p_{ih}\Big(\alpha_i+\bbeta_{i}^{'}{{\bx}_{i}} \Big) \le p_{ih}\Big(\alpha_h+\bbeta _h^{'}\bx_i \Big)  &{\quad}& \forall i,h  \notag\\
	& \bbeta_i\ge \bzero &{\quad}& \forall i  \notag
\end{alignat}
and isotonic CER \eqref{eq:icer} is defined as
\begin{alignat}{2}
	\underset{\alpha,\mathbf{\bbeta}}{\mathop{\min }}&\,\tilde{\tau} \sum\limits_{i=1}^{n}{\rho_{\tilde{\tau}}(y_i - \alpha_i - \bbeta_{i}^{'}{{\bx}_{i}})^2} \label{a2}\\
	\mbox{\textit{s.t.}}\quad
	& p_{ih}\Big(\alpha_i+\bbeta_{i}^{'}{{\bx}_{i}} \Big) \le p_{ih}\Big(\alpha_h+\bbeta _h^{'}\bx_i \Big)  &{\quad}& \forall i,h  \notag\\
	& \bbeta_i\ge \bzero &{\quad}& \forall i  \notag
\end{alignat}

If $p_{ih}=1$, then isotonic CQR and isotonic CER are reduced to the original CQR and CER problem. Therefore, the quantile property (i.e., part i) in Theorem \ref{the:the1}) and the expectile property (i.e., Theorem \ref{the:the2}) are obviously retained.

If $p_{ih}=0$, then \eqref{a1} and \eqref{a2} are simplified as 
\begin{alignat}{2}
	\underset{\alpha,\mathbf{\bbeta}}{\mathop{\min }}&\,\tau \sum\limits_{i=1}^{n}{\rho_\tau(y_i - \alpha_i - \bbeta_{i}^{'}{{\bx}_{i}})} \label{a3}\\
	\mbox{\textit{s.t.}}\quad
	& \bbeta_i\ge \bzero &{\quad}& \forall i  \notag
\end{alignat}
and 
\begin{alignat}{2}
	\underset{\alpha,\mathbf{\bbeta}}{\mathop{\min }}&\,\tilde{\tau} \sum\limits_{i=1}^{n}{\rho_{\tilde{\tau}}(y_i - \alpha_i - \bbeta_{i}^{'}{{\bx}_{i}})^2} \label{a4}\\
	\mbox{\textit{s.t.}}\quad
	& \bbeta_i\ge \bzero &{\quad}& \forall i  \notag
\end{alignat}
Following \citeasnoun{Wang2014c}, the quantile property in isotonic CQR \eqref{a3} is then easy to be verified due to that the proof relies on decision variables $\alpha_i$ only (see proof of Theorem 1 in \citename{Wang2014c}, \citeyear*{Wang2014c}). Similarly, the expectile property for isotonic CER \eqref{a4} is also straightforward in analogy to \citeasnoun{Kuosmanen2020b}. 

\newpage
\subsection{Additional experiments and results}\label{app:experiments} 
\subsubsection{Experiment with different error specifications}\label{app:err}

\begin{table}[H]
  \centering
  \caption{Performance in estimating the quantile function $Q_y$ when $\varepsilon_i=v_i$ and $\varepsilon_i=-u_i$ with $n = 1000$ and $d= 1$, respectively. ICQR = Isotonic CQR, ICER = Isotonic CER, COA = Convexified order-$\alpha$.}
    \begin{tabular}{rrrrrrrrrrr}
    \toprule
    \multicolumn{1}{c}{\multirow{2}[4]{*}{$\varepsilon$}} &       &       & \multicolumn{2}{c}{ICQR} &       & \multicolumn{2}{c}{ICER} &       & \multicolumn{2}{c}{COA} \\
\cmidrule{4-5}\cmidrule{7-8}\cmidrule{10-11}          &       & \multicolumn{1}{l}{$\tau$} & \multicolumn{1}{l}{Bias} & \multicolumn{1}{l}{MSE} &       & \multicolumn{1}{l}{Bias} & \multicolumn{1}{l}{MSE} &       & \multicolumn{1}{l}{Bias} & \multicolumn{1}{l}{MSE} \\
    \midrule
    \multicolumn{1}{l}{$\sigma_v$} & 0.708 & 0.1   & 0.040 & 0.032 &        & 0.048 & 0.029 &       & -1.900 & 4.895 \\
          &       & 0.3   & 0.016 & 0.021 &        & 0.017 & 0.016 &        & -1.593 & 3.352 \\
          &       & 0.5   & 0.017 & 0.020 &        & 0.002 & 0.014 &        & -1.310 & 2.191 \\
          &       & 0.7   & -0.005 & 0.021 &       & -0.014 & 0.016 &       & -1.053 & 1.358 \\
          &       & 0.9   & -0.029 & 0.031 &       & -0.047 & 0.029 &       & -0.800 & 0.761 \\
          & 0.801 & 0.1   & 0.041 & 0.038 &        & 0.057 & 0.035 &        & -1.856 & 4.687 \\
          &       & 0.3   & 0.018 & 0.025 &        & 0.017 & 0.020 &        & -1.571 & 3.257 \\
          &       & 0.5   & 0.017 & 0.023 &        & 0.002 & 0.017 &        & -1.315 & 2.206 \\
          &       & 0.7   & -0.006 & 0.025 &       & -0.014 & 0.019 &       & -1.082 & 1.435 \\
          &       & 0.9   & -0.030 & 0.037 &       & -0.049 & 0.035 &       & -0.843 & 0.855 \\
          & 0.894 & 0.1   & 0.044 & 0.044 &        & 0.052 & 0.041 &        & -1.814 & 4.499 \\
          &       & 0.3   & 0.019 & 0.030 &        & 0.017 & 0.023 &        & -1.553 & 3.177 \\
          &       & 0.5   & 0.019 & 0.028 &        & 0.003 & 0.020 &        & -1.319 & 2.220 \\
          &       & 0.7   & -0.005 & 0.029 &       & -0.015 & 0.023 &       & -1.107 & 1.508 \\
          &       & 0.9   & -0.032 & 0.043 &       & -0.051 & 0.041 &       & -0.884 & 0.955 \\
          &   \vspace{-0.7em}\\  
    \multicolumn{1}{l}{$\sigma_u$} & 1.174 & 0.1   & 0.044 & 0.050 &       & 0.055 & 0.045 &       & -1.670 & 4.468 \\
          &       & 0.3   & 0.009 & 0.028 &        & 0.012 & 0.019 &       & -1.639 & 3.582 \\
          &       & 0.5   & 0.009 & 0.020 &        & -0.005 & 0.013 &      & -1.427 & 2.551 \\
          &       & 0.7   & -0.015 & 0.015 &       & -0.022 & 0.009 &      & -1.157 & 1.587 \\
          &       & 0.9   & -0.040 & 0.010 &       & -0.048 & 0.008 &      & -0.773 & 0.748 \\
          & 0.994 & 0.1   & 0.041 & 0.039 &        & 0.051 & 0.035 &       & -1.818 & 4.810 \\
          &       & 0.3   & 0.009 & 0.022 &        & 0.011 & 0.015 &       & -1.674 & 3.694 \\
          &       & 0.5   & 0.007 & 0.016 &        & -0.006 & 0.010 &      & -1.411 & 2.493 \\
          &       & 0.7   & -0.014 & 0.012 &       & -0.020 & 0.007 &      & -1.104 & 1.450 \\
          &       & 0.9   & -0.037 & 0.008 &       & -0.044 & 0.007 &      & -0.695 & 0.580 \\
          & 0.742 & 0.1   & 0.037 & 0.026 &        & 0.045 & 0.023 &       & -1.983 & 5.371 \\
          &       & 0.3   & 0.008 & 0.015 &        & 0.010 & 0.010 &       & -1.717 & 3.860 \\
          &       & 0.5   & 0.006 & 0.010 &        & -0.005 & 0.007 &      & -1.383 & 2.403 \\
          &       & 0.7   & -0.013 & 0.008 &       & -0.019 & 0.005 &      & -1.028 & 1.264 \\
          &       & 0.9   & -0.034 & 0.006 &       & -0.040 & 0.005 &      & -0.598 & 0.414 \\
    \bottomrule
    \end{tabular}%
  \label{tab:a3}%
\end{table}%

\newpage
\subsubsection{Experiment with model misspecification}\label{app:miss}

\begin{table}[H]
  \centering
  \caption{Performance in estimating quantile function $Q_y$ over noncovex set.}
   \begin{threeparttable}
    \begin{tabular}{rrrrrrrrrr}
    \toprule
    \multicolumn{1}{c}{\multirow{2}[4]{*}{$n$}} & \multicolumn{1}{c}{\multirow{2}[4]{*}{($\sigma^2, \lambda$)}} & \multicolumn{2}{c}{ICQR} &       & \multicolumn{2}{c}{ICER} &       & \multicolumn{2}{c}{COA} \\
\cmidrule{3-4}\cmidrule{6-7}\cmidrule{9-10}          &       & \multicolumn{1}{l}{Bias} & \multicolumn{1}{l}{MSE} &       & \multicolumn{1}{l}{Bias} & \multicolumn{1}{l}{MSE} &       & \multicolumn{1}{l}{Bias} & \multicolumn{1}{l}{MSE} \\
    \midrule
    50   & (1.88, 1.66)  & 0.075 & 0.401 &       & -0.028 & 0.353 &       & -4.510 & 31.867 \\
         & (1.63, 1.24)  & 0.087 & 0.405 &       & -0.016 & 0.355 &       & -4.485 & 31.640 \\
         & (1.35, 0.83)  & 0.096 & 0.411 &       & -0.008 & 0.355 &       & -4.463 & 31.448 \\
    100  & (1.88, 1.66)  & 0.067 & 0.270 &       & -0.022 & 0.228 &       & -4.637 & 32.881 \\
         & (1.63, 1.24)  & 0.076 & 0.273 &       & -0.014 & 0.229 &       & -4.614 & 32.639 \\
         & (1.35, 0.83)  & 0.082 & 0.277 &       & -0.008 & 0.229 &       & -4.593 & 32.432 \\
    200  & (1.88, 1.66)  & 0.056 & 0.180 &       & -0.010 & 0.147 &       & -4.691 & 33.218 \\
         & (1.63, 1.24)  & 0.062 & 0.182 &       & -0.004 & 0.148 &       & -4.668 & 32.975 \\
         & (1.35, 0.83)  & 0.069 & 0.185 &       & 0.000  & 0.148 &       & -4.648 & 32.779 \\
    500  & (1.88, 1.66)  & 0.036 & 0.100 &       & -0.009 & 0.079 &       & -4.733 & 33.649 \\
         & (1.63, 1.24)  & 0.040 & 0.101 &       & -0.005 & 0.080 &       & -4.710 & 33.400 \\
         & (1.35, 0.83)  & 0.044 & 0.102 &       & -0.002 & 0.080 &       & -4.690 & 33.192 \\
    1000 & (1.88, 1.66)  & 0.024 & 0.064 &       & -0.007 & 0.049 &       & -4.750 & 33.751 \\
         & (1.63, 1.24)  & 0.027 & 0.064 &       & -0.004 & 0.050 &       & -4.728 & 33.507 \\
         & (1.35, 0.83)  & 0.029 & 0.065 &       & -0.002 & 0.050 &       & -4.709 & 33.307 \\
    \bottomrule
    \end{tabular}%
    \begin{tablenotes}
        \setlength\labelsep{0pt}
        \footnotesize
        \item DGP: $y_i= x_{i} + 0.1x_{i}^2 + v_i-u_i$, where $x_{i} \sim U[1,10]$, $v_i \overset{\text{i.i.d.}}{\sim} N(0, \sigma_v^2)$, and $u_i \overset{\text{i.i.d.}}{\sim} N^+(0, \sigma_u^2)$.
    \end{tablenotes}
    \end{threeparttable}
  \label{tab:a5}%
\end{table}%

\newpage
\subsubsection{Experiment with outliers}\label{app:out}

\begin{table}[H]
  \centering
  \caption{Performance in estimating the quantile function $Q_y$ with three outliers.}
  \begin{threeparttable}
    \begin{tabular}{ccrrrrrrrrr}
    \toprule
    \multicolumn{2}{c}{\multirow{2}[4]{*}{($\sigma^2, \lambda$)}} & \multicolumn{1}{c}{\multirow{2}[4]{*}{$d$}} & \multicolumn{2}{c}{ICQR} &       & \multicolumn{2}{c}{ICER} &       & \multicolumn{2}{c}{COA} \\
\cmidrule{4-5}\cmidrule{7-8}\cmidrule{10-11}    \multicolumn{2}{c}{} &       & \multicolumn{1}{l}{Bias} & \multicolumn{1}{l}{MSE} &       & \multicolumn{1}{l}{Bias} & \multicolumn{1}{l}{MSE} &       & \multicolumn{1}{l}{Bias} & \multicolumn{1}{l}{MSE} \\
    \midrule
    \multicolumn{2}{c}{(1.88, 1.66)} & 1 & 0.029 & 0.104 &       & -0.011 & 0.081 &       & -1.844 & 19.823 \\
    \multicolumn{2}{c}{} & 2         & 0.032 & 0.105 &       & -0.007 & 0.082 &       & -1.828 & 19.748 \\
    \multicolumn{2}{c}{} & 3         & 0.034 & 0.105 &       & -0.006 & 0.082 &       & -1.814 & 19.686 \\
    \multicolumn{2}{c}{(1.63,  1.24)}& 1 & 0.035 & 0.232 &       & -0.021 & 0.202 &       & -1.760 & 19.840 \\
    \multicolumn{2}{c}{} & 2         & 0.044 & 0.236 &       & -0.012 & 0.203 &       & -1.743 & 19.781 \\
    \multicolumn{2}{c}{} & 3         & 0.050 & 0.238 &       & -0.005 & 0.203 &       & -1.727 & 19.732 \\
    \multicolumn{2}{c}{(1.35, 0.83)} & 1 & 0.011 & 0.408 &       & -0.033 & 0.370 &       & -1.618 & 19.691 \\
    \multicolumn{2}{c}{} & 2         & 0.024 & 0.411 &       & -0.018 & 0.371 &       & -1.598 & 19.639 \\
    \multicolumn{2}{c}{} & 3         & 0.035 & 0.414 &       & -0.006 & 0.372 &       & -1.582 & 19.592 \\
    \bottomrule
    \end{tabular}%
    \begin{tablenotes}
        \setlength\labelsep{0pt}
        \footnotesize
        \item DGP: $y_i = \prod_{d=1}^{D}\BX^{\frac{0.8}{d}}_{d,i} + v_i - u_i$, where $\BX=(\bx1, \bx2)^{'}$, $\bx1_{m} \sim U[1,10]$ ($m=1,\cdots, 200$), $\bx2_{n} \sim U[90,100]$ ($n=1,\cdots, 3$), $v_i \overset{\text{i.i.d.}}{\sim} N(0, \sigma_v^2)$, and $u_i \overset{\text{i.i.d.}}{\sim} N^+(0, \sigma_u^2)$.
        \end{tablenotes}
    \end{threeparttable}
  \label{tab:a7}%
\end{table}%

\end{document}